\documentclass[a4paper,onecolumn,superscriptaddress,nofootinbib,notitlepage,11pt]{article}
\usepackage{jheppub}

 \pdfoutput=1
 \usepackage[utf8]{inputenc}
 \usepackage[english]{babel}
 \usepackage[T1]{fontenc}
 \usepackage{amsmath,amsthm}
 \usepackage{hyperref}
 
 \usepackage{tikz}
 \usepackage{lipsum}           
            
\usepackage{subfig}        

\usepackage[numbers,sort&compress]{natbib} 
\usepackage{url}

\usepackage{multirow}

\usepackage[T1]{fontenc} 

\usepackage{float, array, xspace, amscd, amsthm}
\usepackage{fancyhdr}
\usepackage{longtable}

\newtheorem{proposition}{Proposition}
\newtheorem{definition}{Definition}

\newcommand{\IC}{\mathbb{C}}

\newcommand{\IZ}{\mathbb{Z}}

 \newcommand{\Irr} {\operatorname{Irr}}

\newcommand{\be}{\begin{equation}}
\newcommand{\ee}{\end{equation}}

\usepackage{tikz}
\usetikzlibrary{positioning,intersections}
\usetikzlibrary{calc}
\usetikzlibrary{shapes.geometric}
\usepackage{braids}
\usepackage{tikz-cd}

\newif\ifcomments
\commentstrue

\newif\ifdetails
\detailstrue

\usepackage{upgreek}
\usepackage[normalem]{ulem}
\usepackage{mathtools}
\usepackage{datetime}
\usepackage{cancel}

\usepackage[all,cmtip,knot]{xy}

\newcommand{\orcid}[1]{\href{https://orcid.org/#1}{\includegraphics[width=8pt]{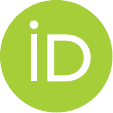}}}

\usepackage{amsmath,amsfonts,amsthm}

\theoremstyle{definition}

\theoremstyle{remark}

\usepackage{bm,latexsym,galois,euscript,dsfont}
\newcommand\EA{\EuScript{A}}
\newcommand\EB{\EuScript{B}}
\newcommand\EC{\EuScript{C}}
\newcommand\ED{\EuScript{D}}

\newcommand\EF{\EuScript{F}}
\newcommand\EM{\EuScript{M}}

\newcommand\EP{\EuScript{P}}
\newcommand\EQ{\EuScript{Q}}
\newcommand\ES{\EuScript{S}}
\newcommand\ER{\EuScript{R}}
\newcommand\EW{\EuScript{W}}

\newcommand{\AN}{\EA(N)^{\times}_{\IZ_2}}

\newcommand{\LMod}{\mathsf{Mod}}
\newcommand{\Hilb}{\mathsf{Hilb}}
\newcommand\Fun{\textsf{Fun}}
\newcommand\Vect{\textsf{Vect}}

\newcommand\Rep {\textsf{Rep}}

\newcommand\id {\mathrm{id}}
\newcommand\Hom {\mathrm{Hom}}
\newcommand{\one}{\mathds{1}}

\newcommand\Mod{\textsf{Mod}}

\newcommand{\Zp}{\mathbb{Z}_p}

\newcommand{\FPdim}{\operatorname{FPdim}}
\newcommand{\Aut}{\operatorname{Aut}}

\settimeformat{ampmtime}

\title{Electric-magnetic duality and $\mathbb{Z}_2$ symmetry enriched Abelian lattice gauge theory}

\author[a,b]{Zhian Jia\orcid{0000-0001-8588-173X},}
\author[a,b]{Dagomir Kaszlikowski,}
\author[c,d]{and Sheng Tan\orcid{0009-0008-3318-9942}}
\affiliation[a]{Centre for Quantum Technologies, National University of Singapore, Singapore 117543, Singapore}
\affiliation[b]{Department of Physics, National University of Singapore, Singapore 117543, Singapore}
\affiliation[c]{Beijing Institute of Mathematical Sciences and Applications, Beijing, 101408, China}
\affiliation[d]{Yau Mathematical Sciences Center, Tsinghua University, Beijing, 100084, China}

\emailAdd{giannjia@foxmail.com}
\emailAdd{phykd@nus.edu.sg}
\emailAdd{tan296@bimsa.cn}

\abstract{
Kitaev's quantum double model is a lattice realization of  Dijkgraaf-Witten topological quantum field theory (TQFT). Its topologically protected ground state space has broad applications for topological quantum computation and topological quantum memory. 
We investigate the  $\mathbb{Z}_2$ symmetry enriched generalization of the model for the Abelian group in a categorical framework and present an explicit Hamiltonian construction.
This model provides a lattice realization of the $\mathbb{Z}_2$ symmetry of the topological phase. We discuss in detail the categorical symmetry of the phase, for which the electric-magnetic (EM) duality symmetry is a special case.
The aspects of symmetry defects are investigated using the $G$-crossed unitary braided fusion category (UBFC).  By determining the corresponding anyon condensation, the gapped boundaries and boundary-bulk duality are also investigated. 
In the last part, an explicit lattice realization of EM duality is discussed.

}

\begin{document}

\maketitle

\section{Introduction}
\label{sec:intro}

Duality is one of the central topics in modern physics. It weaves together different theories such that these theories can be understood from each other's perspective; typical examples include AdS/CFT duality \cite{Maldacena1999,witten1998anti}, Electric-magnetic duality \cite{montonen1977magnetic,kapustin2007electricmagnetic}, T-duality \cite{palan1987duality}, and so on.
For Kitaev's quantum double model $D(G)$ based on a finite group $G$ \cite{Kitaev2003}, which can be regarded as a lattice gauge theoretic realization of Dijkgraaf-Witten topological quantum field theory (TQFT)  \cite{dijkgraaf1990topological}, there is a famous duality known as electric-magnetic (EM) duality \cite{Buerschaper2009mapping,buerschaper2013electric,Hu2018full,wang2020electric,hu2020electric,delcamp2021tensor}.
When the group $G$ is Abelian, the theory has EM duality.
The dual gauge group is $\hat{G}=\mathrm{Irr}(G)$ whose elements are unitary irreducible representations of $G$.
However, when $G$ is non-Abelian, there is no such EM duality. To remedy this problem, the models are generalized to quantum double models based on finite-dimensional semisimple Hopf algebras, for which EM duality can be realized \cite{Buerschaper2013a,buerschaper2013electric}.

From a topologically ordered phase perspective, a quantum double model is a lattice realization of the $2d$ non-chiral topological order \cite{Kitaev2003, Levin2005, KITAEV2006,jia2023weak,jia2023boundary}.
The $2d$ topologically ordered phase is mathematically characterized by a unitary modular tensor category (UMTC) $\ED=(\ED,\oplus,\otimes,\one,\alpha,\beta,\theta)$ \cite{bakalov2001lectures,turaev2016quantum,etingof2016tensor}, where $\alpha,\beta,\theta$ are associator, braiding, and twist.
The anyons are labeled with simple objects $a\in \Irr(\ED)$. $\Irr(\ED)$ is a finite set, and the vacuum charge will be labeled as $\mathds{1}$ (also denoted $0$ when appearing in the index).
Each simple object (anyon) in $\ED$ has an  antiparticle, \emph{viz.}, there is an involution map $*: \Irr(\ED)\to \Irr(\ED)$ such that  $a^{**}=a$ and $\one^*=\one$.
Anyons in $\ED$ can fuse: $a\times b =\sum_{c}N_{ab}^cc$, and the fusion is associative in the sense that $\sum_eN_{ab}^eN_{ec}^d=\sum_{f}N_{bc}^fN_{af}^d$. The numbers $\{N_{ab}^c\}$ are called the fusion multiplicities and satisfy the symmetric properties: $N_{ab}^c=N_{ba}^c=N_{a^*b^*}^{c^*}=N_{ac^*}^{b^*}$; and $N_{ab^*}^0=\delta_{ab}$.
Anyon's braiding is characterized by $\beta: a\otimes b\to b\otimes a$ and each anyon has its topological spin $\theta_a$.
Each anyon $a\in\ED$ has its quantum dimension $d_a$ for which $d_a=d_{a^*}$. From fusion process we have $d_ad_b=\sum_cN_{ab}^cd_c$ (it is evident that $d_a$ is the eigenvalue of the matrix $[N_a]_{bc}=N^c_{ab}$). The anyons with quantum dimension $1$ are called Abelian anyons, otherwise they will be called non-Abelian anyons. 
If all anyons of a phase are Abelian, the phase is called an Abelian topological phase.
The typical Abelian topological phases are Kitaev's $\IZ_N$ quantum double $D(\mathbb{Z}_N)$ phase, whose topological order is given by the representation category $\mathsf{Rep}(D(\IZ_N))$.
These Abelian phases can also be realized using Levin-Wen's string-net model by choosing input unitary fusion category $\EC$ as the representation category $\Rep (\IZ_N)$ or the category of $\IZ$-graded vector spaces $\Vect_{\IZ_N}$ \cite{Levin2005}. These different realizations are related with each other by Morita equivalence  $\mathsf{Rep}(D(\IZ_N)) \simeq  \mathcal{Z}(\Rep( \IZ_N)) \simeq   \mathcal{Z}(\Vect_{\IZ_N)}) $ \cite{etingof2016tensor}, where ``$\simeq$'' represents braided equivalence and $\mathcal{Z}$ represents taking the Drinfeld center.

In this work, we consider another kind of generalization, that is, the topological phase enriched with a symmetry group. This kind of phase will be called the symmetry enriched topological (SET) phase.
By introducing symmetry, the corresponding phase will exhibit more complicated quasi-particle excitations and mutual statistics. 
The SET phases have been explored using modular extension of a unitary braided fusion category (UBFC) \cite{lan2017modular} and using $G$-crossed UBFC \cite{cui2016gauging,Barkeshli2019symmetry,bischoff2019spontaneous}.
Besides their fundamental importance, the study of these SET phases and topological defects is also very crucial for topological quantum computation.
Two typical examples are: 
(i) A topologically protected degeneracy appears when we cut boundaries of $\IZ_N$ toric code on a sphere. This topological degeneracy generates an error correcting code \cite{Dennis2002topological, Terhal2015quantum};  (ii) Twist defects fix the problem of Abelian anyons of $\IZ_N$ toric code, converting it into a universal topological quantum computer \cite{Nayak2008, Cong2017universal}.

In this paper, we study the simplest SET phase: the $\IZ_2$ symmetry enriched Abelian topological phase.
In Sec.~\ref{sec:ZN}, we review the Hamiltonian of Kitaev quantum double model for abelian group and discuss its topological properties.
Then in Sec.~\ref{sec:SET}, we review the SET phase in $G$-crossed UBFC framework. The EM duality symmetry is a special case of the more general categorical symmetry. We will discuss the EM duality in detail within the framework of categorical symmetry. 
One of our main contributions is given in Sec.~\ref{sec:SETbd}, where we study the gapped boundary of the SET phase from the perspective of anyon condensation. 
The gapped boundary phase is obtained from anyon condensation of the SET phase, and the bulk phase is recovered from the boundary phase by taking relative center, and different boundary phases are Morita equivalent.
We review the existing framework of SET boundary theory, which is based on condensable algebra description, and propose a new formalism grounded on Frobenius algebra.
With the above preparation, in Sec.~\ref{sec:latticeReal}, we present a lattice realization of the EM duality symmetry via the ribbon operator and discuss its physical properties.

\section{Quantum double model for Abelian topological order}
\label{sec:ZN}

\subsection{$\mathbb{Z}_N$ model}
Let us first recall the construction of the $D(G)$ quantum double model for a finite group $G$.
For a given lattice embedded in a surface $\Sigma=V(\Sigma)\cup E(\Sigma)\cup F(\Sigma)$, where $V(\Sigma)$, $E(\Sigma)$ and  $F(\Sigma)$ are vertex set, edge set and face set, respectively, we assign a Hilbert space $\mathcal{H}_e=\mathbb{C}[G]$ to each edge $e\in E(\Sigma)$. The total Hilbert space is then given by $\mathcal{H}(\Sigma)=\otimes_{e\in E(\Sigma)} \mathcal{H}_e$.
We then introduce two kinds of local stabilizers. The first one is the projective operator which is attached to each vertex $v\in V(\Sigma)$:
\begin{equation}\label{eq:Av}
	A_v
\big{|}	\begin{aligned}
		\begin{tikzpicture}
			\draw[-latex,black] (-0.5,0) -- (0,0);
			\draw[-latex,black] (0,0) -- (0,0.5); 
			\draw[-latex,black] (0,0) -- (0.5,0); 
			\draw[-latex,black] (0,-0.5) -- (0,0); 
			\draw [fill = black] (0,0) circle (1.2pt);
			\node[ line width=0.2pt, dashed, draw opacity=0.5] (a) at (0.7,0){$x_1$};
			\node[ line width=0.2pt, dashed, draw opacity=0.5] (a) at (-0.7,0){$x_3$};
			\node[ line width=0.2pt, dashed, draw opacity=0.5] (a) at (0,-0.7){$x_4$};
			\node[ line width=0.2pt, dashed, draw opacity=0.5] (a) at (0,0.7){$x_2$};
		\end{tikzpicture}
	\end{aligned}   \big{ \rangle}     
= 
\frac{1}{|G|} \sum_{g\in G}
\big{|}	\begin{aligned}
	\begin{tikzpicture}
		\draw[-latex,black] (-0.5,0) -- (0,0);
		\draw[-latex,black] (0,0) -- (0,0.5); 
		\draw[-latex,black] (0,0) -- (0.5,0); 
		\draw[-latex,black] (0,-0.5) -- (0,0); 
		\draw [fill = black] (0,0) circle (1.2pt);
		\node[ line width=0.2pt, dashed, draw opacity=0.5] (a) at (1,0){$x_1g^{-1}$};
		\node[ line width=0.2pt, dashed, draw opacity=0.5] (a) at (-1,0){$gx_3$};
		\node[ line width=0.2pt, dashed, draw opacity=0.5] (a) at (0,-0.7){$gx_4$};
		\node[ line width=0.2pt, dashed, draw opacity=0.5] (a) at (0,0.75){$x_2g^{-1}$};
	\end{tikzpicture}
\end{aligned}   \big{ \rangle} .
\end{equation}
The vertex operator can be regarded as a gauge transformation averaged over $G$, which imposes a local Gauss constraint. The second one is the local projector associated to
each face $f\in F(\Sigma)$: 
\begin{equation}\label{eq:Bf}
	B_f
	\big{|}	\begin{aligned}
		\begin{tikzpicture}
			\draw[-latex,black] (-0.5,0.5) -- (0.5,0.5);
			\draw[-latex,black] (-0.5,-0.5) -- (-0.5,0.5); 
			\draw[-latex,black] (0.5,-0.5) -- (0.5,0.5); 
			\draw[-latex,black] (-0.5,-0.5) -- (0.5,-0.5); 
			\draw [fill = black] (0,0) circle (1.2pt);
			\node[ line width=0.2pt, dashed, draw opacity=0.5] (a) at (0.75,0){$x_1$};
			\node[ line width=0.2pt, dashed, draw opacity=0.5] (a) at (-0.75,0){$x_3$};
			\node[ line width=0.2pt, dashed, draw opacity=0.5] (a) at (0,-0.7){$x_4$};
			\node[ line width=0.2pt, dashed, draw opacity=0.5] (a) at (0,0.7){$x_2$};
		\end{tikzpicture}
	\end{aligned}   \big{ \rangle}     
	= 
	\delta_{x_1^{-1}x_2x_3x_4^{-1}, e}
	\big{|}	\begin{aligned}
	\begin{tikzpicture}
		\draw[-latex,black] (-0.5,0.5) -- (0.5,0.5);
		\draw[-latex,black] (-0.5,-0.5) -- (-0.5,0.5); 
		\draw[-latex,black] (0.5,-0.5) -- (0.5,0.5); 
		\draw[-latex,black] (-0.5,-0.5) -- (0.5,-0.5); 
		\draw [fill = black] (0,0) circle (1.2pt);
		\node[ line width=0.2pt, dashed, draw opacity=0.5] (a) at (0.75,0){$x_1$};
		\node[ line width=0.2pt, dashed, draw opacity=0.5] (a) at (-0.75,0){$x_3$};
		\node[ line width=0.2pt, dashed, draw opacity=0.5] (a) at (0,-0.7){$x_4$};
		\node[ line width=0.2pt, dashed, draw opacity=0.5] (a) at (0,0.7){$x_2$};
	\end{tikzpicture}
\end{aligned}   \big{ \rangle} .
\end{equation}
Here $e$ is the identity element of $G$, and $\delta$ satisfies $\delta_{x,y} = 1$ if $x=y$ and $0$ otherwise. The face operator imposes a local flatness condition.
The Hamiltonian is given by
\begin{equation}\label{eq:QDhamiltonian}
	H=-\sum_{v\in V(\Sigma)} A_v-\sum_{f\in F(\Sigma)} B_f.
\end{equation}
The topological excitations of the model  are classified by pairs $a_{[g], \pi}:=([g], \pi)$ where $[g]$ is a conjugacy class of the group $G$, and $\pi$ is an irreducible representation (irrep) of the centralizer $C_{G}(g)$. 
The vacuum charge corresponds to $g=e$ and $\pi=\mathds{1}$ (the trivial representation).
The antiparticle is given by
$([g^{-1}],\pi^{\dagger})$ (note that $C_G(g)=C_G(g^{-1})$). 
The conjugacy class $[g]$ of a topological charge is called the magnetic charge and the irrep $\pi$ is called the electric charge. When $g=e$,  $a_{[e],\pi}$ is characterized by a representation of $G$ and is called a chargeon; when  $\pi=\mathds{1}$, $a_{[g],\mathds{1}}$ is called a fluxion; and  when both $g\neq e$ and $\pi\neq \mathds{1}$, $a_{[g],\pi}$ is called a dyon.
The quantum dimension (i.e.~Frobenius-Perron dimension) of the topological excitation $a_{[g], \pi}$ is given by
\begin{equation}\label{eq:dim}
	\FPdim a_{[g],\pi} =|[g]| \dim \pi.
\end{equation}
These topological excitations form a UMTC $\mathsf{Rep} (D(G))$, the representation category of the quantum double of the finite group $G$.

When $G=\IZ_N$, we introduce the generalized Pauli operators (a.k.a., discrete Weyl operators, or Heisenberg-Weyl operators)
\begin{equation}\label{eq:ZXN}
	X_N=\sum_{h\in \IZ_N} |h+1\rangle \langle h|,\quad Z_N=\sum_{h\in \IZ_N} \omega_N^h |h\rangle \langle h|,
\end{equation}
where $\omega_N=e^{2\pi i/N}$ is the $N$-th root of unity.
They satisfy the commutation relation 
\begin{equation}
	Z_NX_N =\omega_N X_NZ_N.
\end{equation}
For any choice of $(h,g)\in \IZ_N\times \IZ_N$, we define
\begin{equation}
	Y^{g,h}=X^g_NZ^h_N.
\end{equation}
From the commutation relation of $X_N$ and $Z_N$, we obtain
\begin{equation}
Y^{g,h}	Y^{k,l}=\omega_N^{hk-gl} Y^{k,l}Y^{g,h}.
\end{equation}

The eigenstates of $Z_N$ are $|h\rangle$ with corresponding eigenvalues $\omega_N^h$, and the eigenstates of $X_N$ are
\begin{equation}\label{eq:Xeigen}
	|\psi_g\rangle =\frac{1}{\sqrt{N}} \sum_{h\in \IZ_N} \omega_N^{gh}  |h\rangle
\end{equation}
with the corresponding eigenvalues $\omega_N^{-g}$.
Using the notations in Ref.~\cite{Kitaev2003}, we have
\begin{equation}\label{eq:L-X}
    L_+^g|x\rangle=X_N^g|x\rangle,\quad  L_-^g|x\rangle =(X^{\dagger}_N)^g|x\rangle.
\end{equation}
And also $Z_N^g|x\rangle=\omega_N^{gx}|x\rangle$, $(Z^{\dagger}_N)^g|x\rangle=\omega_N^{-gx}|x\rangle$.
Hereinafter, to avoid cluttering of equations, we will denote $X_N$ and $Z_N$ simply by $X$ and $Z$ whenever there is no ambiguity.

The quantum double model can be constructed using these operators.
For a given edge direction's configuration of a vertex and a face (e.g., in Eqs.~(\ref{eq:Av}) and (\ref{eq:Bf})), we define
\begin{equation}
	A(v)=X_1^{\dagger} X_2^{\dagger}X_3X_4,\quad B(f)=Z_1^{\dagger} Z_2Z_3Z_4^{\dagger}.
\end{equation}
They are unitary but not Hermitian in general. The Hamiltonian in Eq.~(\ref{eq:QDhamiltonian}) becomes
\begin{equation}
	H=-\sum_{v\in V(\Sigma)} \frac{1}{N} \sum_{h\in \IZ_N} A(v)^h-\sum_{f\in F(\Sigma)}  \frac{1}{N} \sum_{h\in \IZ_N} B(f)^h,
\end{equation}
where $A_v=  \frac{1}{N}\sum_{h\in \IZ_N} A(v)^h$ and $B_f= \frac{1}{N} \sum_{h\in \IZ_N} B(f)^h$ are local stabilizers. In fact, the first identity follows from Eq.~\eqref{eq:L-X} immediately, and the second follows from 
\begin{equation}
\begin{aligned}
    \frac{1}{N}\sum_{h\in\mathbb{Z}_N} B(f)^h &\big{|}	\begin{aligned}
	\begin{tikzpicture}
		\draw[-latex,black] (-0.5,0.5) -- (0.5,0.5);
		\draw[-latex,black] (-0.5,-0.5) -- (-0.5,0.5); 
		\draw[-latex,black] (0.5,-0.5) -- (0.5,0.5); 
		\draw[-latex,black] (-0.5,-0.5) -- (0.5,-0.5); 
		\draw [fill = black] (0,0) circle (1.2pt);
		\node[ line width=0.2pt, dashed, draw opacity=0.5] (a) at (0.75,0){$x_1$};
		\node[ line width=0.2pt, dashed, draw opacity=0.5] (a) at (-0.75,0){$x_3$};
		\node[ line width=0.2pt, dashed, draw opacity=0.5] (a) at (0,-0.7){$x_4$};
		\node[ line width=0.2pt, dashed, draw opacity=0.5] (a) at (0,0.7){$x_2$};
	\end{tikzpicture}
    \end{aligned}   \big{ \rangle}   = \frac{1}{N}\sum_{h\in\mathbb{Z}_N}\omega_N^{(-x_1+x_2+x_3-x_4)h}\big{|}	\begin{aligned}
	\begin{tikzpicture}
		\draw[-latex,black] (-0.5,0.5) -- (0.5,0.5);
		\draw[-latex,black] (-0.5,-0.5) -- (-0.5,0.5); 
		\draw[-latex,black] (0.5,-0.5) -- (0.5,0.5); 
		\draw[-latex,black] (-0.5,-0.5) -- (0.5,-0.5); 
		\draw [fill = black] (0,0) circle (1.2pt);
		\node[ line width=0.2pt, dashed, draw opacity=0.5] (a) at (0.75,0){$x_1$};
		\node[ line width=0.2pt, dashed, draw opacity=0.5] (a) at (-0.75,0){$x_3$};
		\node[ line width=0.2pt, dashed, draw opacity=0.5] (a) at (0,-0.7){$x_4$};
		\node[ line width=0.2pt, dashed, draw opacity=0.5] (a) at (0,0.7){$x_2$};
	\end{tikzpicture}
\end{aligned}   \big{ \rangle} \\
& = \delta_{-x_1+x_2+x_3-x_4,0} \big{|}	\begin{aligned}
	\begin{tikzpicture}
		\draw[-latex,black] (-0.5,0.5) -- (0.5,0.5);
		\draw[-latex,black] (-0.5,-0.5) -- (-0.5,0.5); 
		\draw[-latex,black] (0.5,-0.5) -- (0.5,0.5); 
		\draw[-latex,black] (-0.5,-0.5) -- (0.5,-0.5); 
		\draw [fill = black] (0,0) circle (1.2pt);
		\node[ line width=0.2pt, dashed, draw opacity=0.5] (a) at (0.75,0){$x_1$};
		\node[ line width=0.2pt, dashed, draw opacity=0.5] (a) at (-0.75,0){$x_3$};
		\node[ line width=0.2pt, dashed, draw opacity=0.5] (a) at (0,-0.7){$x_4$};
		\node[ line width=0.2pt, dashed, draw opacity=0.5] (a) at (0,0.7){$x_2$};
	\end{tikzpicture}
\end{aligned}   \big{ \rangle} = B_f\big{|}	\begin{aligned}
	\begin{tikzpicture}
		\draw[-latex,black] (-0.5,0.5) -- (0.5,0.5);
		\draw[-latex,black] (-0.5,-0.5) -- (-0.5,0.5); 
		\draw[-latex,black] (0.5,-0.5) -- (0.5,0.5); 
		\draw[-latex,black] (-0.5,-0.5) -- (0.5,-0.5); 
		\draw [fill = black] (0,0) circle (1.2pt);
		\node[ line width=0.2pt, dashed, draw opacity=0.5] (a) at (0.75,0){$x_1$};
		\node[ line width=0.2pt, dashed, draw opacity=0.5] (a) at (-0.75,0){$x_3$};
		\node[ line width=0.2pt, dashed, draw opacity=0.5] (a) at (0,-0.7){$x_4$};
		\node[ line width=0.2pt, dashed, draw opacity=0.5] (a) at (0,0.7){$x_2$};
	\end{tikzpicture}
\end{aligned}   \big{ \rangle}. 
\end{aligned}
\end{equation}
Note that we have used the identity $(1/N)\sum_{h\in\mathbb{Z}_N}\omega_N^{gh} = \delta_{g,0}$, $\forall\,g\in\mathbb{Z}_N$. 

Notice that for $\mathbb{Z}_2$ case (toric code), one usually shifts the local terms by a constant, namely,
    \begin{align}
        A_v=I+A(v)=I+X\otimes X\otimes X\otimes X,\quad 
        B_f=I+B(f)=I+Z\otimes Z\otimes Z\otimes Z. 
    \end{align}
But we will take $A(v)=X\otimes X\otimes X\otimes X$ and $B(f)=Z\otimes Z\otimes Z\otimes Z$ as local stabilizers.
This point of view can also be generalized to the $\mathbb{Z}_N$ case.

\begin{proposition}\label{prop:GSspace}
    The ground state space of the $\mathbb{Z}_N$ quantum double model is given by
    \begin{equation}
        \mathcal{V}_{\rm GS}=\{\,|\psi\rangle\,|\, A(v)|\psi\rangle =|\psi\rangle, B(f)|\psi\rangle =|\psi\rangle,\forall\, v,f
        \}.
    \end{equation}
\end{proposition}
\begin{proof}
For vertex stabilizers, it is clear that $A(v)|\psi\rangle =|\psi\rangle$ implies $A_v|\psi\rangle =|\psi\rangle$.
The reverse direction can be obtained by taking the spectral decomposition $A(v)=\sum_{\lambda} e^{i \theta_{\lambda}}  \Pi_{\lambda}$, where $\theta_{\lambda}= 2 \pi\lambda/N$ ($\lambda\in \mathbb{Z}_N$), and $\Pi_{\lambda}$ is the projection to the corresponding eigenspace of $A(v)$. 
This implies that
\begin{equation}
    A_v=\frac{1}{N}\sum_{g\in \mathbb{Z}_N}A(v)^g=\sum_{\lambda} \left(\frac{1}{N}\sum_{g\in \mathbb{Z}_N} e^{ig \theta_{\lambda}}\right)\Pi_{\lambda}=\Pi_{0}.
\end{equation}
From this we see that $A_v|\psi\rangle=|\psi\rangle$ implies that $A(v)|\psi\rangle=|\psi\rangle$.
For face stabilizers, we can similarly show that conditions $B(f)|\psi\rangle=|\psi\rangle$ and $B_f|\psi\rangle=|\psi\rangle$ are equivalent.
\end{proof}

\vspace{1em}
\emph{Ribbon operator and excitations.} ---
Notice that each group element of $\IZ_N$ is a conjugacy class $[g]=\{g\}$ and all their centralizers are the Abelian group $\IZ_N$, thus $N$ irreps $\pi$ are all one-dimensional.
There are $N^2$ different anyons, from Eq.~(\ref{eq:dim}), all of these anyons have quantum dimension 1,  thus the resulting topological phase $\EA(N)=\Rep (D(\IZ_N))$ is an Abelian phase:
\begin{equation}
\EA(N)=\{	\one, e^{g}, m^{h}, \varepsilon^{g,h}=e^g\otimes m^h \}.
\end{equation}
We will also use the notion $\varepsilon^{g,0}=e^g$,  $\varepsilon^{0,h}=m^h$ and $\varepsilon^{0,0}=\one$.
The electric charge $e^g$  is associated with a vertex $v$ such that $A(v) |\psi\rangle =\omega_{N}^g|\psi\rangle$; the magnetic charge $m^h$ is associated with a face $f$ such that $B(f) |\psi\rangle =\omega_N^h|\psi\rangle $; and dyon $\varepsilon^{g,h}$ is associated with a site $s=(v,f)$ such that both of the above conditions hold.

Since $\IZ_N$ is Abelian, the ribbon operator $F_{\rho}^{g,h}$ over a ribbon $\rho$ becomes two string operators, $F^{g,h}_{\rho}=(\otimes_{i\in \partial \rho} Z(i)^g )\otimes (\otimes_{j\in \bar{\partial} \rho} X(j)^{h})$,
where $\partial \rho$ and $\bar{\partial} \rho$ respectively denote the direct boundary and dual boundary of $\rho$.
Notice that here the string operators are constructed in a way such that they are consistent with the definition of vertex operators and face operators; for example, a $Z^g$-string operator along a direct path $\partial \rho=e_1e_2\cdots e_n$ is defined as 
\begin{equation}
Z^g(\partial \rho)=\otimes_{i\in \partial \rho} Z(i)^g,
\end{equation}
where we set $Z(i)=Z$ when the direction of the edge $e_i$ is the same as the direction of the path, otherwise we set $Z(i)=Z^{\dagger}$. Similarly for $X^h$-strings (the direction of the dual edge is defined as the direction obtained by rotating $\pi/2$ of the direct edge counterclockwise).
A $Z^g$-string creates $e^g$ and $e^{-g}$ particles over two ends of the string. Similarly, an $X^h$-string creates $m^h$ and $m^{-h}$ particles over two ends of the string.

The fusion rule is given by
\begin{equation}
	(e^g\otimes m^{h})\otimes (e^k\otimes m^{l})=e^{g+k}\otimes m^{h+l}.
\end{equation}
The fusion rule is invariant under an exchanging of electric and magnetic charges $e \leftrightarrow m$.
Different anyons can also braid with each other, and both fusion and braiding can be realized in the lattice using string operators.

Although we will mainly study the quantum double lattice realization of Abelian phase $\EA(N)$, it is also helpful to investigate the phase from the string-net perspective in some situations (especially when in the category-theoretic framework).
In this case, the input data is a UFC $\EC=\Rep(\IZ_N)$, and the topological excitation is characterized by its Drinfeld center $\mathcal{Z}(\EC)$.
In \cite{Kitaev2012a}, it is pointed out that the topological excitations can also be regarded as point defects in the bulk.
A point defect is a zero-dimensional defect between two bulk trivial domain walls. 
It is mathematically described by a bimodule functor between two $\EC|\EC$-bimodules, thus $\mathcal{Z}(\EC)\simeq \Fun_{\EC|\EC} (\EC,\EC)$ where  $\Fun_{\EC|\EC} (\EC,\EC)$
is the category of all bimodule functors.
To summarize, the cyclic Abelian topological phase is characterized by UMTCs
\begin{equation}
\begin{aligned}
\EA(N)&=\Rep(D(\IZ_N))\simeq \mathcal{Z}(\Rep(\IZ_N)) \simeq \mathcal{Z}(\Vect(\IZ_N)) \\
&\simeq \Fun_{\Rep(\IZ_N) | \Rep(\IZ_N)} (\Rep(\IZ_N),\Rep(\IZ_N)) \\
&\simeq \Fun_{\Vect_{\IZ_N} | \Vect_{\IZ_N}} (\Vect_{\IZ_N},\Vect_{\IZ_N}). 
\end{aligned}
\end{equation}
These different realizations of the same topological phase are Morita equivalent.

\subsection{General Abelian group model}

Let us now extend the $2d$ cyclic group model to the Abelian group case.
According to the fundamental theorem of Abelian groups, a finite Abelian group $G$ can be decomposed as 
\begin{equation}
    G \cong \mathbb{Z}_{N_1}\otimes \cdots \otimes \mathbb{Z}_{N_n},
\end{equation}
for which $|G|=p_1^{\alpha_1}\cdots p_n^{\alpha_n}$ and $N_i=p_i^{\alpha_i}$ ($p_i$'s are prime numbers).
We denote $g\in G$ as an $n$-tuple $g=(g_1,\cdots,g_n)$, where $g_i\in \mathbb{Z}_{N_i}$.
All irreducible representations of $G$ are one-dimensional; we introduce their corresponding characters as
\begin{equation}
    \chi_g(k)=\exp\left(2\pi i \sum_{j=1}^n \frac{g_jk_j}{N_j}\right).
\end{equation}
Notice that $\chi_g(a+b)=\chi_g(a)\chi_g(b)$ and $\chi_g\chi_h=\chi_{g+h}$.
The set of all characters forms a dual group $G^{\vee}$ which is isomorphic to $G$.

To each edge of the lattice, we attach a Hilbert space $\mathcal{H}_e=\mathbb{C}[G]\cong\mathbb{C}^{N_1}\otimes \cdots \otimes \mathbb{C}^{N_n}$, where for $g\in G$, $|g\rangle=|g_1\rangle \otimes \cdots \otimes |g_n\rangle$.
We introduce the following operators
\begin{equation}
    X_G^g=\sum_{h\in G} |h+g\rangle \langle h|, \quad Z_G^g=\sum_{h\in G} \chi_g(h) |h\rangle \langle h|.
\end{equation}
If we denote $X_{N_i}$ and $Z_{N_i}$ the operators introduced in Eq.~\eqref{eq:ZXN} for cyclic group $\mathbb{Z}_{N_i}$, it is clear that 
\begin{equation}
    X_G=\otimes_{i=1}^n X_{N_i}, \quad Z_G=\otimes_{i=1}^n Z_{N_i}.
\end{equation}
and $ X_G^g=(X_G)^g=\otimes_{i=1}^n X_{N_i}^{g_i}$, $Z_G^g=(Z_G)^g=\otimes_{i=1}^n Z_{N_i}^{g_i}$.

The vertex operator is of the form
\begin{equation}
    A(v)=\prod_{j\in \partial v} X_G(j), \quad A_v=\frac{1}{|G|}\sum_{g\in G}A(v)^g,
\end{equation}
and the face operator is of the form 
\begin{equation}
    B(f)=\prod_{j\in \partial f}Z_G(j), \quad B_f=\frac{1}{|G|}\sum_{g\in G} B(f)^g.
\end{equation}
Notice that $\chi_g(h)=\chi_h(g)$ and thus that $\sum_{g\in G} \chi_{g}(h_1)\cdots \chi_g(h_n)=\sum_{g\in G}\chi_{h_1+\cdots +h_n}(g)=\delta_{h_1+\cdots +h_n,e}$ (here $e=0$).
Operators $A(v)$ and $B(f)$ are unitary but not Hermitian; their eigenvalues are $\omega_G^g=\exp(2 \pi i \sum_{j=1}^n\frac{g_j}{N_j})$.
The ground state manifold is the invariant state space of $A(v)$ and $B(f)$.

\begin{proposition}\label{prop:GSspace1}
    The ground state space of the Abelian group quantum double model is given by
    \begin{equation}
        \mathcal{V}_{\rm GS}=\{\,|\psi\rangle\,|\, A(v)|\psi\rangle =|\psi\rangle,\, B(f)|\psi\rangle =|\psi\rangle,\forall\, v,f
        \}.
    \end{equation}
\end{proposition}
This can be proved in a similar way as the cyclic group case.

\section{Algebraic theory of $\mathbb{Z}_2$ symmetry enriched Abelian topological order }
\label{sec:SET}

Before discussing the $\mathbb{Z}_2$ enriched Abelian topological order, let us first review the general algebraic theory of SET.
The main object we are going to discuss is symmetry defect, which is a point-like object but not intrinsic topological excitation of the topological order.
It plays an important role in our understanding of symmetries of the topological order.

\subsection{Categorical symmetry of topological phase}

We first review the definition of categorical symmetry of the topological phase \cite{cui2016gauging,Barkeshli2019symmetry,bischoff2019spontaneous}, which is universal, namely, it is independent of the specific Hamiltonian realization of the phase.

To motivate the definition, let us recall that for the Ising phase, which is an equivalence class of Hamiltonians $\mathbf{Ising}=\{ (H_{\mathrm{Ising}},\mathcal{H}_{\mathrm{Ising}})\}$.
For a $k$-site Hamiltonian $H_k$, the full symmetry of the phase is $\mathcal{U}_k$, which is a set of unitary (or antiunitary) operators $U_k\in B( \mathcal{H}_k),\mathcal{H}_k\in \mathbf{Ising}$ such that $U_kH_k=H_kU_k$.
By saying that the phase has
a $G$ symmetry,  we mean that for each $k\in \IZ_+$, there is a homomorphism $\rho_k: G\to \mathcal{U}_k$. 
The above discussion of symmetry is based on the Hamiltonian realization of the phase. To introduce the universal symmetry of the phase, which does not depend on the Hamiltonian realization, we need to introduce the concept of categorification of a group $G$ (also known as $2$-group).

The categorification of $G$ is denoted as $\underline{G}$, whose objects are elements of $G$ and the $\Hom(g,h)=\delta_{g,h}$, i.e., there is no morphism between different elements and for $g$ there is only one morphism to itself (identity morphism), namely, $G$ is regarded as a discrete category. We can define the tensor product for the category $\underline{G}$ as $g\otimes h:=gh$, 
the tensor unit is thus the identity element $e=1\in G$.
Equipped with this tensor product, $\underline{G}$ becomes a strict monoidal category.
Here we also introduce some other crucial algebraic notions that will be used later in our discussion:
\begin{itemize}
	\item We denote $\operatorname{Aut}_{\otimes}^{br}(\ED)$ the group of braided monoidal autoequivalences of a UMTC $\ED$ up to natural isomorphisms, and $\operatorname{Aut}_{\otimes}(\EC)$ the group of monoidal autoequivalences of a UFC $\EC$ up to natural isomorphisms.
	\item Let $\operatorname{Inv}(\EC)$ be the group of invertible objects in a UFC $\EC$ with group operation given by tensor product and the unit given by tensor unit. 
	\item The Brauer-Picard group $\operatorname{BrPic}(\EC)$ of a UFC $\EC$ is the group of equivalence classes of invertible $\EC$-bimodule categories with respect to the $\EC$-module tensor product for the UFC $\EC$,
	and we have the group isomorphism $\operatorname{BrPic}(\EC)\simeq \operatorname{Aut}_{\otimes}^{br}(\mathcal{Z}(\EC))$ and monoidal equivalence
	$\underline{\operatorname{BrPic}(\EC)}\simeq \underline{\operatorname{Aut}_{\otimes}^{br}(\mathcal{Z}(\EC))}$.
	\item The Picard group $\operatorname{Pic}(\ED)$ of UMTC $\ED$ is the group of invertible $\ED$-modules (which can be regarded as $\ED$-bimodules using the braiding of $\ED$). There are a group isomorphism ${\operatorname{Pic}(\ED)}\simeq {\operatorname{Aut}_{\otimes}^{br}(\ED)}$ and a monoidal equivalence $\underline{\operatorname{Pic}(\ED)}\simeq \underline{\operatorname{Aut}_{\otimes}^{br}(\ED)}$.	
\end{itemize}

For topological phase UMTC $\ED$, the full symmetry is characterized by $\operatorname{Aut}_{\otimes}^{br}(\ED)$.
When the phase is non-chiral, $\ED=\mathcal{Z}(\EC)$ for a UFC $\EC$, then the full symmetry is equivalently characterized by the Brauer-Picard group of $\EC$, Picard group of $\ED$, and braided monoidal equivalence group of $\ED$, \emph{viz}.,
 \begin{equation}
 \operatorname{Aut}_{\otimes}^{br}(\ED) \simeq \operatorname{BrPic}(\EC) \simeq \operatorname{Pic}(\ED).
 \end{equation}
A functor $\varphi \in \operatorname{Aut}_{\otimes}^{br}(\ED)$ can be regarded as a symmetry of the topological phase, it may permute the topological charge of the phase $\varphi (a)=a'$ but the vacuum charge must be left invariant $\varphi(\one)=\one$.
The gauge invariant information of the phase must be left invariant under the symmetry transformation, thus we have
\begin{align}
	&\text{fusion multiplicity:}\quad N_{ab}^c=N_{a'b'}^{c'};\\
	&\text{quantum dimension:}\quad \FPdim a =\FPdim a';	\\
	&F\text{-matrix}: \quad F_{a}^{bcd}=F_{a'}^{b'c'd'};\\
	&\text{braiding:} \quad R_{ab}^c=R_{a'b'}^{c'};\\
	&\text{topological spin:}\quad  \theta_{a}=\theta_{a'};\\
	&\text{modular matrix:} \quad S_{ab}=S_{a'b'}.
\end{align}
A $G$ symmetry of $\ED$ is a group homomorphism $\rho:G\to \operatorname{Aut}_{\otimes}^{br}(\ED)$.
If $\rho$ can be promoted into a monoidal functor $\underline{\rho}:\underline{G}\to \underline{\operatorname{Aut}_{\otimes}^{br}(\ED)}$, it is referred to as a categorical global symmetry of $\ED$. 
Technically, it consists of the following data: 
\begin{itemize}
	\item[(a)] unitary braided monoidal functors   $(\underline{\rho}_g, s_g)_{g\in G}$ where $\underline{\rho}_g:=\underline{\rho}(g)$ and $s_g^{X,Y}:\underline{\rho}_g(X \otimes Y)\overset{\simeq}{\to} \underline{\rho}_g(X)\otimes \underline{\rho}_g(Y) $ and $s_g^{\one}:\underline{\rho}_g(\one) \overset{\simeq}{\to} \one$ are isomorphisms. Here ``braided monoidal functor'' means that the functor preserves the tensor product and braiding;
	\item[(b)] natural isomorphisms $\gamma_{g,h}:\underline{\rho}_{gh}\overset{\simeq}{\to} \underline{\rho}_g\comp \underline{\rho}_h$ and $\gamma_e:\underline{\rho}_{e}\overset{\simeq}{\to}\id_{\ED}$. For an anyon $a\in \ED$, we will denote the $g$ action simply by $g(a):=\underline{\rho}_g(a)$ when there is no ambiguity. Then, in terms of objects, we have $gh(a)\overset{\simeq}{\to} g(h(a))$ and $e(a)\overset{\simeq}{\to}a$.	
\end{itemize}

The action of global symmetry on simple objects of $\ED$ is just like a permutation, namely, for $a\in \Irr(\ED)$, $g(a)\in \Irr(\ED)$. In particular, $g(\one)=\one$ for all $g\in G$.
The topological phase $\ED$ is left invariant under the categorical global symmetry.

\subsection{Symmetry defects and SET phase} \label{sec:sym-defects}

For a topological phase $\ED$ with a categorical symmetry $G$, we can introduce the point-like defects carrying flux associated with symmetry $g\in G$, such kind of defects are called symmetry defects.

Since symmetry defects are extrinsic objects of a topological phase $\ED$, to describe it, extra data are needed. For a given symmetry group $G$ of the topological phase $\ED$, it turns out that a symmetry defect carrying $g$-flux can be modeled with simple objects in a $\ED$-bimodule category $\ED_g$.
We refer to the collection $\ED_g$ of symmetric defects with $g$-flux as $g$-defect sector and denote objects in $\ED_g$ as $a_g$ with a $g$-subscript, and $\ED_0=\ED$ (here we use $0$ to denote identity of $G$ instead of $1_G$) is called the trivial sector.
Since symmetric defect behaves like a point-particle, anyons in $\ED$ can fuse with
defects, and defects in different sectors can also fuse, this means that $\oplus_{g\in G}\ED_g$ forms a multi-fusion category. We can also consider the braiding of defects and anyons, these data give a $G$-crossed UBFC structure.

In summary, the bulk phase with $G$-symmetry defects are characterized by a  $G$-crossed UBFC $\EP=\ED_G^{\times}$ which is a  $G$-crossed  unitary braided extension of $\ED$ with respect to the categorical symmetry $\underline{\rho}:\underline{G}\to \underline{\Aut_{\otimes}^{\mathrm{br}}(\ED)} $; it consists of the following data:

(a) a (not necessarily faithful\footnote{A faithful $G$-grading is a grading $\oplus_{g\in G}\ED_g$ for which each sector $\ED_g\neq 0$.}) $G$-grading $\EP=\ED_G^{\times}=\oplus_{g\in G}\ED_g$, and the fusion is $G$-crossed fusion, namely, for $a_g\in \ED_g,b_h\in \ED_h$, $a_g\otimes b_h\in \ED_{gh}$ for all $g,h\in G$;

(b) a $G$-action $\underline{\rho}: \underline{G}\to \underline{\operatorname{Aut}_{\otimes}(\EP)}$ extending the action on $\ED_0=\ED$, which satisfies $g(\ED_h)=\ED_{ghg^{-1}}$ for all $g,h\in G$. The associator for $\underline{\rho}_g$ is denoted as $s_g$ and the associator between $\underline{\rho}_{gh}$ and $\underline{\rho}_g\comp \underline{\rho}_h$ is denoted as $\gamma_{g,h}$;

(c) a $G$-crossed braiding which extends the braiding on $\ED_0=\ED$. It is a set of natural isomorphisms $c_{a_g,b}:a_g\otimes b\to g(b)\otimes a_g$ with $a_g\in \ED_g$ and $b\in \EP$, graphically
\begin{equation}
	\begin{aligned}
		\begin{tikzpicture}
			\braid[
			line width=1.5pt,
			style strands={1}{red},
			style strands={2}{blue}] (Kevin)
			s_1^{-1} ;
			\node[label = above:$a_g$] at (Kevin-2-s) {};
			\node[at= (Kevin-1-s),label = above:$g(b)$]  {};
			\node[label = : $a_g$] at (Kevin-2-s) {};
			\node[at=(Kevin-1-e),label = below: $b$] {};
			\node[at=(Kevin-2-e),label = below: $a_g$] {};
		\end{tikzpicture}
	\end{aligned}. \label{eq:Gbraiding}
\end{equation}
These data are required to satisfy certain consistency conditions.
We refer the interested readers to \cite{drinfeld2010braided,turaev2010homotopy,etingof2010fusion} for details of the definition.
The phase $\EP=\ED_{G}^{\times}$ is usually referred to as the SET phase.

It is natural to ask that for a given topological phase $\ED$ and global $G$ symmetry ${\rho}:G\to \operatorname{Aut}_{\otimes}^{br}(\ED)$, if such extension exists. This is studied extensively in \cite{etingof2010fusion,cui2016gauging,Barkeshli2019symmetry}.
To introduce the symmetric defect and give a well-defined fusion rule, the first obstruction is the cohomology class $O_3(\rho)\in H^3(G,\operatorname{Inv}(\ED))$ where $\operatorname{Inv}(\ED)$ is the group of invertible objects in $\ED$. This is the same as the obstruction to lifting group homomorphism $\rho$ to a categorical group homomorphism. If this obstruction vanishes, the consistent fusion rules of defects can be defined but the fusion rule is not unique, they are parameterized by cohomology classes $\alpha \in  H^2(G,\operatorname{Inv}(\ED))$. Now when we have the consistent fusion rule, which is given by a pair $(\rho,\alpha)$, the fusion operation may not be associative, this is the obstruction $O_4(\rho,\alpha)\in H^4(G,U(1))$. When the obstruction vanishes, the resulting fusion category is parameterized by $\beta \in H^3(G,U(1))$. See \cite{etingof2010fusion,cui2016gauging,Barkeshli2019symmetry} for details.

There are two useful properties about the bulk phase with symmetry defects $\ED_G^{\times}=\oplus_{g\in G} \ED_g$ that we will use later:
\begin{itemize}
	\item The number of simple objects in $\ED_g$ (called rank of $\ED_g$) is equal to the number of fixed points of the action $\rho_g$ on $\Irr(\ED)$ (see \cite{Barkeshli2019symmetry}).
	\item The Frobenius-Perron dimensions of different sectors of $\ED_G^{\times}$ are equal, that is $\operatorname{FPdim}(\ED_g)=\operatorname{FPdim}(\ED_h)$ for all $g,h\in G$ (see \cite[Proposition 3.23]{MUGER2004galois}).
\end{itemize}

For $G$-enriched SET phase $\EP=\ED_G^{\times}$, it is called non-chiral if there exist a UFC $\EC$ and $G$-graded UFC $\EC_G^{\times}$ such that $\EP=\mathcal{Z}_{\EC}(\EC_G^{\times})\simeq \oplus_{g\in G}\mathcal{Z}_{\EC}(\EC_g)$ and $\ED_g=\mathcal{Z}_{\EC}(\EC_g)$, otherwise it is called chiral.
Here $\mathcal{Z}_{\EC}$ is a relative center concerning $\EC$.
The SET phases we will discuss later are all assumed to be non-chiral.

\subsection{Symmetry-gauged SET phase}
\label{sec:subSET}
The symmetry defects with flux $g\in G$ are extrinsic defects,
realized microscopically by deforming the uniform Hamiltonian $H_0$. When we regard the symmetry defects as dynamical objects, we can discuss their confinement and deconfinement. If the energy cost to separate two symmetry defects grows with the distance between them, they are confined, meanwhile, if the energy cost to separate them is finite and independent of their distance, they are deconfined.
When the $G$-symmetry defects are deconfined, the corresponding global $G$-symmetry of the phase effectively becomes an emergent local gauge invariance at the long-wavelength limit; this process is generally known as symmetry-gauging \cite{cui2016gauging,Barkeshli2019symmetry}.
For a SET phase $\EP=\ED_G^{\times}$, after gauging the symmetry group $G$, the symmetry defects become deconfined quasiparticles of the gauged phase and the resultant phase is described by a UMTC $\EP^G=(\ED_G^{\times})^G$.

Mathematically, the symmetry gauging of the SET phase is characterized by the \emph{equivariantization} of the $G$-crossed UBFC $\ED_G^{\times}$. 

For a given $G$-crossed UBFC $\ED_G^{\times}$, and a subgroup $H$ of $G$, an \emph{$H$-equivariant object} in $\ED_G^{\times}$ is a pair $(X, \{u_h\}_{h\in H})$ consisting of an object $X\in \ED_G^{\times}$ and a collection of isomorphisms $u_h: \rho_h(X)\overset{\simeq}{\to} X$ for all $h\in H$ such that the following diagram commutes for all $h\in H$:
	\begin{equation}
		\begin{tikzcd}
			{g}(h(X)) \arrow[r,black, "g(u_h)"] \arrow[d, black,"\gamma_{g,h}^X"] & g(X) \arrow[d, "u_g" black] \\
			gh(X) \arrow[r, black, "u_{gh}" black]
			& X
		\end{tikzcd}	
	\end{equation}
	An \emph{$H$-equivariant morphism} between $H$-equivariant objects $(X, \{u_h\}_{h\in H})$ and $(Y, \{v_h\}_{h\in H})$ is a morphism $f\in \Hom_{\ED_G^{\times}}(X,Y)$ such that the following diagram commutes:
	\begin{equation}
		\begin{tikzcd}
			h(X) \arrow[r,black, "u_h"] \arrow[d, black,"h(f)"] & X \arrow[d, "f" black] \\
			h(Y) \arrow[r, black, "v_{h}" black]
			& Y
		\end{tikzcd}	
	\end{equation}
	The tensor product between objects $(X, \{u_h\}_{h\in H})$ and $(Y, \{v_h\}_{h\in H})$ is defined as
	\begin{equation}
		(X, \{u_h\}_{h\in H})\otimes (Y, \{v_h\}_{h\in H})= (X\otimes Y, \{(u_h\otimes v_h)\comp s_h^{X,Y}\}_{h\in H})	
	\end{equation}
	where $s_g^{X,Y}:h(X\otimes Y)\to h(X)\otimes h(Y)$ is the associator for the functor $\underline{\rho}_h$. 
	The tensor unit is $(\one_{\ED},t_g: g(\one_{\ED})\to \one_{\ED})$.
	All $H$-equivariant objects in $\ED_G^{\times}$ with $H$-equivariant morphisms form a fusion category, which we refer to as \emph{$H$-equivariantization of $\ED_G^{\times}$} and denote it as $(\ED_G^{\times})^H$. When $H=G$, $(\ED_G^{\times})^G$ is also known as \emph{equivariantization of $\ED_G^{\times}$}.

In \cite{kirillov2008g}, it is proven that for $G$-crossed extension of UMTC $\ED$, the equivariantization $(\ED_G^{\times})^G$ is a UMTC.
Therefore, after gauging the symmetry, the resulting SET phase can be understood just like the usual anyon model.
Via the symmetry gauge, the $G$-crossed UBFC description of the SET phase can be related to the modular extension description  \cite{lan2017modular,lan2018classification}.

\subsection{EM duality symmetry enriched  Abelian SET phase}

We are now in a position to discuss the EM duality enriched SET phase for Abelian quantum double phase $\EA(N)$. 
As shown in \cite{teo2015theory,Barkeshli2019symmetry}, these SET phases are characterized by the relative center of Tambara-Yamagami categories \cite{tambara1998tensor}, thus, they are non-chiral SET phases. In this part, we will present a detailed discussion of these phases and their lattice realization. In the next section, we will show that these Tambara-Yamagami categories are the boundary phases.

For Abelian phase $\EA(N)$, consider the EM duality $\IZ_2=\{0,1\}$ symmetry 
 \begin{equation}
 	\underline{\rho}:\underline{\IZ_2}\to \underline{\Aut_{\otimes}^{br} (\EA(N))},
 \end{equation}	
 for which $\underline{\rho}_0=\id$ and $\underline{\rho}_1$ permutes electric and magnetic charges $e\leftrightarrow m$ (this  $\underline{\rho}_1$ is the EM duality symmetry).
The anyon sector is just the $\IZ_N$ anyons $\EA(N)_0=\Rep(D(\IZ_N))$.
For $1$-flux symmetry defect sector  $\EA(N)_1$, 
notice that there are $N$ topological charges that is invariant under the action of $\underline{\rho}_1$; they are $\varepsilon^{g,g}$ for all $g\in \IZ_N$. As we have pointed at the end of Subsection \ref{sec:sym-defects}, $\EA(N)_1$ has $N$ simple objects, which we denote as  $\sigma_0,\cdots,\sigma_{N-1}$. The corresponding SET phase thus is
\begin{equation}
	\EA(N)_{\mathbb{Z}_2}^{\times}=\EA(N)_0\oplus \EA(N)_1=\{\one,e^g,m^h,\varepsilon^{g,h}\}\oplus \{\sigma_0,\cdots,\sigma_{N-1}\}.
\end{equation}
This is the relative center of the Tambara-Yamagami category for cyclic group $\IZ_N$.

Tambara and Yamagami completely classified the $\IZ_2$-graded fusion categories in which all but one element are invertible, and the only non-invertible simple object carries a nontrivial symmetry flux. For $\IZ_N$ case, 
\begin{equation}
	\mathsf{TY}(\IZ_N)= 	\mathsf{TY}(\IZ_N)_0\oplus 	\mathsf{TY}(\IZ_N)_1=\{a|a\in\IZ_N\} \oplus \{\sigma\}.
\end{equation}
The fusion rule is given by
\begin{equation}
	a\otimes b =a\oplus b, \quad a\otimes \sigma= \sigma\otimes a=\sigma,\quad \sigma\otimes \sigma =\oplus_{a\in \IZ_N} a.
\end{equation}
The associator $\alpha$ is parameterized by a non-degenerate bicharacter $\chi:\IZ_N\times \IZ_N\to \IC^{\times}$ and a square root $\tau$ of $|\IZ_N|^{-1}=1/N$, and the unit morphisms are the identity map.
$\mathsf{TY}(\IZ_N)$ is rigid, with $a^*=-a$ and $\sigma^*=\sigma$.
The quantum dimensions are $\FPdim a=1$ for all $a\in \IZ_N$ and $\dim \sigma=\sqrt{N}$. See Appendix~\ref{sec:TYcat} for a detailed discussion.

From the relative center construction \cite{gelaki2009centers}, the fusion rules of the SET phase 	$\EA(\mathbb{Z}_N)_{\mathbb{Z}_2}^{\times}$ can be calculated. The fusion between anyons remains unchanged; for fusion involving symmetry defects, we have
\begin{equation}\label{eq:fusionZN}
	\varepsilon^{g,h} \otimes \sigma_k=\sigma_k \otimes 	\varepsilon^{g,h} = \sigma_{k+g-h},\quad
	\sigma_k\otimes \sigma_l=\oplus_{g} \varepsilon^{g+k+l,g}.
\end{equation}
The braiding structure can also be calculated.
The EM duality symmetry operation $\underline{\rho}_1$ exchanges the electric and magnetic charges; however, the symmetry defects are invariant under EM duality symmetry.

\section{Two algebraic theories of gapped boundary}
\label{sec:SETbd}
A crucial feature of topologically ordered phases is the boundary-bulk duality \cite{KONG2017}. In this duality, the boundary phase is obtained from bulk by anyon condensation, and the bulk phase is recovered from the boundary by taking Drinfeld center.
Here we discuss a generalization of the above picture to the case of the SET phase.
We will discuss the boundary theory of SET phase via anyon condensation.
We first review the existing anyon condensation theory based on Lagrangian algebra. Then we propose a new anyon condensation theory based on Frobenius algebra and prove the equivalence of two approaches in describing the gapped boundaries.

\subsection{Bulk-to-boundary condensation I}
Consider a $2d$ (not necessarily non-chiral) phase $\ED$ enriched by a symmetry group $G$, and the corresponding SET phase is described by a  $G$-crossed UBFC $\EP_1=\ED_G^{\times}$.
Suppose that an anyon condensation happens in a region inside the phase $\EP_1$. 
During the condensation process, some anyons are confined in the domain wall and the others become deconfined and can cross the domain wall to enter the condensed phase, and similarly for twist defects.
From this picture, 
if the symmetry group $G$ of the SET phase is broken into a subgroup $K\leq G$, 
the condensed phase is described by another
$K$-crossed UBFC $\EP_2=\EF_K^{\times}$, and the domain wall between them is a $1d$ phase described by a $G$-crossed UFC $\EP_{1|2}=\EW_G^{\times}$.
When the symmetry group $G$ is completely broken and $\EP_2$ becomes a trivial phase, the domain wall becomes a gapped boundary.
In this part, we will discuss the relations between these phases using natural physical requirements.

From the anyon condensation theory between two topological phases $\ED$ and $\EF$ without symmetry, the anyon condensation is controlled by a \emph{condensable algebra} $A\in \ED$ \cite{Kong2014}.

\begin{definition}
A condensable algebra $A$ in a UMTC $\ED$ is an algebra with multiplication $\mu_{A}:A \otimes A\to  A$ and unit map $\eta_{A}: \mathds{1}\to A$ such that: 
\begin{itemize}
\item[(a)] $A$ is commutative, i.e., $\mu_{A} \comp c_{A,A}=\mu_{A}$ where $c_{A,A}$ is the braiding map in $\ED$; 

\item[(b)] $A$ is separable, i.e., $\mu_{A}: A \otimes A \to A$ splits as a morphism of $A$-bimodules. Namely, there is an $A$-bimodule map $e_{A}: A \rightarrow A \otimes A$ such that $\mu_{A} \comp e_{A}=\mathrm{id}_{A}$; 

\item[(c)] $A$ is connected, viz., $\operatorname{dim} \operatorname{Hom}_{\ED}(\mathds{1}, A)=1$. In other words, the vacuum charge $\mathds{1}$ only appears once in the decomposition of $A=\mathds{1}\oplus a\oplus b\oplus \cdots$.
\end{itemize}
If an algebra only satisfies (a) and (b), it is called an \emph{\'{e}tale algebra} \cite{davydov2013witt}. If a condensable algebra $A$ satisfies $(\operatorname{FPdim} A)^2=\operatorname{FPdim} \ED:=\sum_{a\in \mathrm{Irr}(\ED)}(\operatorname{FPdim} a)^2$, then $A$ is called a {Lagrangian algebra}
\end{definition}

The condensable algebra determines the anyon condensation from one phase $\ED$ to another phase $\EuScript{F}$ with domain wall $\EW$. When the phase $\EuScript{F}=\mathsf{Hilb}$ is a trivial topological phase, the gapped domain wall becomes a gapped boundary and the corresponding condensable algebra must be Lagrangian.
In \cite{frohlich2006correspondences}, it is proven that an algebra $A$ is commutative if and only if, when $A$ decomposes into direct sum of simple objects in $\ED$ as $A=\oplus_{a\in \mathrm{Irr}(\ED)} N_{A}^a a$, all $a$'s in the decomposition are bosonic (topological spin $\theta_a=1$) if $N_{A}^a\neq 0$. This means that $A$ is a coherent superposition of bosons. In \cite{Cong2017}, it is shown that $A$ is separable if and only if for any $a,b\in \ED$, there is a partial isometry $\Hom(a,A)\otimes \Hom(b,A)\to \Hom(a\otimes b,A)$; and $(\operatorname{FPdim} A)^2=\operatorname{FPdim} \ED$ if and only if, for all $a,b\in \Irr(\ED)$, we have $N_{A}^aN_{A}^b\leq \sum_{c\in \operatorname{Irr}(\ED)}N_{ab}^cN_{\EA}^c$ where $N_{ab}^c$ are fusion coefficients of $a,b$.
These results can help us to determine the Lagrangian algebras of a given bulk phase $\ED$. 
It is worth mentioning that the algebra object $A$ cannot uniquely determine the Lagrangian algebra. There exist examples of Lagrangian algebras $(A,\mu_{A},\eta_{A})$ and $(A,\mu'_{A},\eta'_{A})$ with the same algebra object but non-isomorphic Lagrangian algebra structures \cite{davydov2014bogomolov}. 

Let us now recall what happens for anyons during the condensation from UMTC $\ED$ to UMTC $\EF$ with a gapped domain wall $\EW$.
The condensation is controlled by a condensable algebra $A\in \ED$, and $A$ is the vacuum (tensor unit) of the condensed phase $\EF$. 
When an anyon $a$ of $\ED$ moves into the vicinity of domain wall between $\ED$ and $\EF$, it will fuse with $A$.
The resulting domain wall phase $\EW$ is thus given by the monoidal functor $ \bullet \otimes A$. These fused anyons form a right $A$-module category: $\EW=\LMod_A(\ED)$. Some of the domain wall anyons are confined in the domain wall, while some others are deconfined that they can cross the wall and become anyons of the condensed phase $\EF$. The deconfined anyons are local $A$-modules in $\EW$, thus $\EF=\LMod_A^{{loc}}(\ED)$.

In order to generalize the above picture for anyon condensation to the case of SET phase, the first thing we need to consider is the condition for a categorical symmetry $\underline{K}\hookrightarrow \underline{G}\to \underline{\Aut_{\otimes}^{br}(\ED)}$ to be preserved after anyon condensation.
This has been extensively studied in \cite{bischoff2019spontaneous}, and it is shown that this is equivalent to two conditions:
\begin{itemize}
\item[(a)] For each $k\in K$, $k(A)$ is isomorphic to $A$ as algebras;
\item[(b)] There is a $K$-equivariant algebra structure $u_K=\{u_k:k(A)\to A\}_{k\in K}$ on $A$.
\end{itemize}
This means that $A$ must be a $K$-equivariant condensable algebra.
Assume that the symmetry group $G$ is preserved during the anyon condensation, \emph{viz.}, the above two conditions of $A$ hold for group $G$. 
The anyon condensation of $\ED_0=\ED$ sector remains unchanged. For $g$-flux symmetry defect sector $\ED_g$, similar to anyons, these $g$-flux defects fuse with $A$ and results in an $A$-module category $\EW_g=\LMod_A(\ED_g)$.
Some $g$-flux symmetry defects are confined in the domain wall and others are deconfined, thus they can cross the domain wall to become $g$-flux defects in the condensed phase. The deconfined defects in $A$-module framework are thus $g$-local modules \cite{bischoff2019spontaneous}: a module $M$ is called $g$-local if $\mu_M=\mu_M\comp (u_{g}\otimes \id_M)\comp c_{M,A}\comp c_{A,M}$. Notice that here the $g$-local module condition is a natural result of the $g$-crossed braiding.
To summarize, the condensed phase is given by
\begin{equation}
	\EF^{\times}_G=\oplus_{g\in G} (\EW_g)_{dc}\simeq \oplus_{g\in G} \LMod_{A}^{g\text{-loc}}(\ED_g).
\end{equation}
Here $\LMod_{A}^{g\text{-loc}}(\ED_g)$ denotes the category of $g$-local modules in $\ED_g$. 

For unitary $G$-crossed braided extension $\ED_G^{\times}$ of anyon model $\ED$, the domain wall phase determined by the condensable algebra $A$ is $\EW_G^{\times}\simeq \LMod_{A}(\ED_{G}^{\times})\simeq \oplus_{g\in G}\LMod_{A}(\ED_{g})$, and the condensation is described as 
	\begin{equation}
		\ED_G^{\times} \overset{\bullet\otimes A}{\longrightarrow}\EW_G^{\times}\simeq \LMod_{A}(\ED_{G}^{\times})\simeq \oplus_{g\in G}\LMod_{A}(\ED_{g}) \to \EF_G^{\times}=\oplus_{g\in G}\LMod^{g\text{-loc}} (\ED_g).
	\end{equation}
When $A$ is a Lagrangian algebra, and all symmetries are broken, $\EF_g=\emptyset$ for all $g\neq 1_G$; so the condensed phase is the trivial phase $\Hilb$.
In this case, we obtain a gapped boundary $\EW_G^{\times}\simeq \LMod_{A}(\ED_{G}^{\times})\simeq \oplus_{g\in G}\LMod_{A}(\ED_{g}) $, which is a $G$-crossed UFC.

\subsection{Bulk-to-boundary condensation II}
Now we give another way to describe the anyon condensation of the SET phase using the quotient category construction, which generalizes the anyon condensation process presented in \cite{Cong2017}. This description is also very powerful for us to investigate the gapped boundaries. 
To this end, we first introduce the notion of special Frobenius algebra.

\begin{definition}
	Let $\ED$ be a tensor category. A Frobenius algebra in $\ED$ is a quintuple $(F, \mu,\eta,\Delta, \varepsilon)$ with $F\in \ED$, $\mu:F\otimes F\to F$, $\eta:\one \to F$, $\Delta :F\to F\otimes F$ and $\varepsilon: F\to \one$, such that  $(F, \mu,\eta)$ forms an algebra, $(F, \Delta, \varepsilon)$ forms a coalgebra, and
	\begin{equation}
		(\id_F\otimes \mu)\comp (\Delta\otimes \id_F)=\Delta\comp\mu= (\mu\otimes \id_F)\comp (\id_F \otimes \Delta).
	\end{equation}
	A Frobenius algebra is called special if there exist $\alpha,\beta \in \mathbb{C}^{\times}$ such that
	\begin{equation}
		\mu \comp \Delta = \alpha  \id_F, \quad \varepsilon \comp \eta =\beta \id_{\one}. 
	\end{equation}
	When $\alpha=1$ and $\beta=\FPdim F$, it is called normalized-special.
\end{definition}
A condensable algebra $(A,\mu_{A},\eta_{A})$ in UMTC $\ED$ is a normalized-special Frobenius algebra, see \cite{Kong2014} for a proof.

\begin{definition}[Pre-quotient category $\ED/F$]
Let $\ED$ be a BFC and $F$ a special Frobenius algebra in $\ED$. The pre-quotient category $\ED/F$ consists of the following data:
	
	(a) Objects $\mathrm{Obj} (\ED/F)=\mathrm{Obj} (\ED)$;
	
	(b) Morphisms are defined as $\Hom_{\ED/F}(X,Y):=\Hom_{\ED}(F\otimes X, Y)$ and the composition between $f\in \Hom_{\ED/F}(X,Y)$ and $g\in \Hom_{\ED/F}(Y,Z)$ is defined as $g\comp f:=  g\comp(\id_F\otimes f)\comp (\Delta \otimes \id_X)$;
	
	(c) Tensor product is defined as $X\otimes_{\ED/F} Y: =X\otimes_{\ED} Y$ and for morphisms $f\in \Hom_{\ED/F}(X,Y)$, $g\in \Hom_{\ED/F}(Z,W)$, $f\otimes_{\ED/F} g:= (f\otimes g)\comp(\id_F\otimes c_{F,X}\otimes \id_Z)\comp (\Delta\otimes \id_X\otimes \id_Z)$.
\end{definition}

The functor $Q: \ED\to \ED/F$, which maps object $X$ to $X$ and morphism $f\in \Hom_{\ED}(X,Y)$ to $\varepsilon \otimes f \in \Hom_{\ED/F}(X,Y)$, is a faithful monoidal functor.
Since we would like to describe the anyon condensation using the pre-quotient category, the problem is that $\ED/F$ may not be semisimple. Thus we need the following definition.

\begin{definition}[Idempotent completion or Karoubi envelope]For the pre-quotient category $\ED/F$, we can define its canonical idempotent completion as the category $\EQ(\ED,F)$ which consists of the following data:
	
	(a) The objects in $\EQ(\ED,F)$ are $\mathrm{Obj}(\EQ(\ED,F))=\{(X,p)\,|\,X\in\mathrm{Obj} (\ED/F), p=p^2\in \operatorname{End}_{\ED/F}(X)\}$. Here $p:F\otimes X\to X$ is known as an idempotent of $X$;
	
	(b) The morphisms of $\EQ(\ED,F)$ are given by 
	$$\Hom_{\EQ(\ED,F)}((X,p),(Y,q)):=\{f\in\Hom_{\ED/F}(X,Y)|f\comp p=q\comp f\}.$$
	
	All other structures of $\EQ(\ED,F)$ are inherited from $\ED/F$.
\end{definition}

The anyon condensation is controlled by the normalized-special Frobenius algebra structure of a condensable algebra $A$ in $\ED$. The vacuum of the gapped domain wall turns out to be $A$ and the wall excitation $\EW$ is the idempotent completion $\EQ(\ED, A)$ of the pre-quotient category $\ED/A$. 
The condensation from phase $\ED$ to domain wall phase $\EW$ is given by the functor
\begin{equation}
	\ED\overset{Q}{\longrightarrow}\ED/A	\overset{I.C.}{\longrightarrow} \EQ(\ED,A),
\end{equation}
where $Q$ and $I.C.$ represents the pre-quotient functor and idempotent complete functor respectively.
Notice that gapped wall phase is a $1d$ phase, thus is described by the UFC $\EW\simeq\EQ(\ED,A)$.
The quantum dimension of the domain wall phase $\EW$ is \cite{MUGER2004galois}
\begin{equation}
	\FPdim \EW =	\FPdim \EQ(\ED,A)= \frac{\FPdim \ED}{\FPdim A}.
\end{equation}
It is proven that the category $\EQ(\ED, A)$ is equivalent to the $A$-module category $\LMod_{A}(\ED)$ in $\ED$ \cite{MUGER2004galois}, thus the domain wall phase can also be described by the category  $\LMod_{A}(\ED)$ and the condensation is given by the functor
\begin{equation}
	F: \ED \to  \LMod_{A}(\ED), \;X\mapsto A\otimes X,
\end{equation}
where $A\otimes X$ automatically has an $A$-module structure originating from the algebra structure of $A$. The adjoint functor of $F$ is the forgetful functor $I:\LMod_{A}(\ED) \to \ED, (M,\mu_M)\mapsto M$. This matches well with the anyon condensation description we have discussed before. 

Let us now see what happens for symmetry defect sector $\ED_g$. In this case, to ensure that the symmetry $K\leq G$ is not broken,  $A$ must have a $K$-equivariant Frobenius structure. 
For simplicity, we assume that the entire symmetry group $G$ is not broken.
Since $X_g\in \ED_g$ behaves like a point particle, when it is dragged near the domain wall, it undergoes the same physical process as the anyons. Therefore, the domain wall defect sector is the idempotent completion $\EQ(\ED_g, A)$ of the pre-quotient category $\ED_g/A$. For $G$-crossed UBFC $\ED_G^{\times}$, the condensable algebra $A$ gives the domain wall phase as $\EW_G^{\times}\simeq\EQ(\ED_G^{\times}, A)\simeq \oplus_{g\in G}\EQ(\ED_g,A)$, and the condensation is described as 
\begin{equation}
    \ED_G^{\times} \overset{Q}{\longrightarrow}\ED_{G}^{\times}/A \simeq \oplus_{g\in G} (\ED_g/A) 	\overset{I.C.}{\longrightarrow} \EW^{\times}_G= \EQ(\ED_G^{\times}, A)\simeq \oplus_{g\in G}\EQ(\ED_g,A).
\end{equation}
For each domain wall $g$-sector $\EW_g$, we have the decomposition
\begin{equation}
	\EW_g=(\EW_g)_{\rm c}\oplus (\EW_g)_{\rm dc},
\end{equation}
where ``c'' and ``dc'' represent confined and deconfined respectively. 
The deconfined $g$-flux defect $X_g$ must have $g$-trivial double braiding with the vacuum $(A, u_G=\{u_g:g(A)\to A\})$ of the condensed phase, 
\begin{equation}
	(u_{g}\otimes \id_{X_g})\comp c_{X_g,A}\comp c_{A,X_g}	=\id_{A\otimes X_g}.
\end{equation}
Graphically, we have
\begin{equation}\label{eq:doublebraiding}
	\begin{aligned}
		\begin{tikzpicture}
			\braid[
			line width=1.5pt,
			style strands={1}{red},
			style strands={2}{blue},	
			border height= 0.05cm] (Kevin)
			s_1^{-1} s_1^{-1} ;
			\node[at= (Kevin-2-s),yshift=0.5cm,label = above: $X_g$]{};
			\node[at= (Kevin-1-s),yshift=0.5cm, label = above: $A$]  {};
			\node[at=(Kevin-1-e),label = below: $A$] {};
			\node[at=(Kevin-2-e),label = below: $X_g$] {};
			\draw (0.75,0) rectangle (1.25,0.4);
			\node (start) [at=(Kevin-1-s),yshift=0.2cm] {$u_g$};
			\draw[red, line width=1.5pt] (1,0.4) -- (1,0.7);
			\draw[blue, line width=1.5pt] (2,0) -- (2,0.7);
		\end{tikzpicture}
	\end{aligned}
	=
	\begin{aligned}
		\begin{tikzpicture}[
			/pgf/braid/.cd,
			style strands={1}{red},
			style strands={2}{blue},
			number of strands=2,
			line width=1.5pt,
			border height= 1cm
			]
			\braid[] (identity)  1 1;
			\node[at= (identity-1-s), label = above: $A$]{};
			\node[at= (identity-2-s),label = above: $X_g$]  {};
			\node[at=(identity-2-e),label = below: $X_g$] {};
			\node[at=(identity-1-e),label = below: $A$] {};
		\end{tikzpicture}
	\end{aligned}.
\end{equation}
Here $g$-trivial double braiding is a direct result of the $G$-crossed braiding structure of the SET phase.
The $g$-sector of the condensed SET phase is then $\EF_g=(\EW_g)_{\rm dc}$. Suppose that $G$ is not broken during the condensation, the condensed phase will be $\EF_{G}^{\times} =\oplus_{g\in G} \EF_g$. When $G$ is completely broken, then the condensed phase is a trivial phase $\Hilb$, and the domain wall becomes a gapped boundary.

The equivalence of this description with the $A$-module description is given by 
\begin{equation}
	\LMod_{A}(\ED_g) \simeq \EQ(\ED_g,A),\quad \LMod_{A}(\ED_G^{
		\times})\simeq \EQ(\ED_G^{
		\times},A).
\end{equation}
We would like to stress that the crossed tensor product between $A \in \ED_0=\ED$ and $X_g\in \ED_g$ ensures that $A\otimes X_g \in \ED_g$, thus $\ED_g/A$ can be well-defined.

\subsection{Example: Gapped boundaries of cyclic Abelian SET phase}

Now let us consider the gapped boundaries of the  Abelian SET phase  $\EA(N)^{\times}_{\IZ_2}$. To simplify the discussion, we assume that $N$ is a prime number $p$.
Recall that $\EA(p)$ is the quantum double phase $\EA(p)=\Rep(D(\IZ_p))$.
In Ref.~\cite{bravyi1998quantum}, it has been argued via relative homology group method that there are two types of boundaries of $\mathbb{Z}_2$-toric code: smooth and rough ones. A more general analysis is given for general $D(G)$-quantum double model for finite group $G$ in Refs.~\cite{Bombin2008family,Beigi2011the,Cong2017}, which indicates that the boundary is determined by a subgroup $K\subseteq G$ up to conjugation and a 2-cocycle $\omega\in H^2(K,\mathbb{C}^{\times})$. Since there are two subgroups $K_1=\mathbb{Z}_1$ and $K_2=\mathbb{Z}_p$ (recall that we assume the order $p$ of the cyclic group $\IZ_p$ is prime) and  the 2nd cohomology group for both cases are trivial, there are exactly two boundaries of the model. The triviality of the 2-cocycle can be seen as follows: 
the case that $H^2(\mathbb{Z}_1,\mathbb{C}^{\times})=0$ is trivial; for $H^2(\mathbb{Z}_p,\mathbb{C}^{\times})$, we use the relation  $H^{n}(G,\mathbb{Z})\simeq H^{n-1}(G,\mathbb{Q}/\mathbb{Z})\simeq H^{n-1}(G,\mathbb{C}^{\times})$ for all $n\geq 2$ and the fact $H^3(\mathbb{Z}_p,\mathbb{Z})=0$. 

\begin{figure}[t]
	\centering
	\includegraphics[width=12cm]{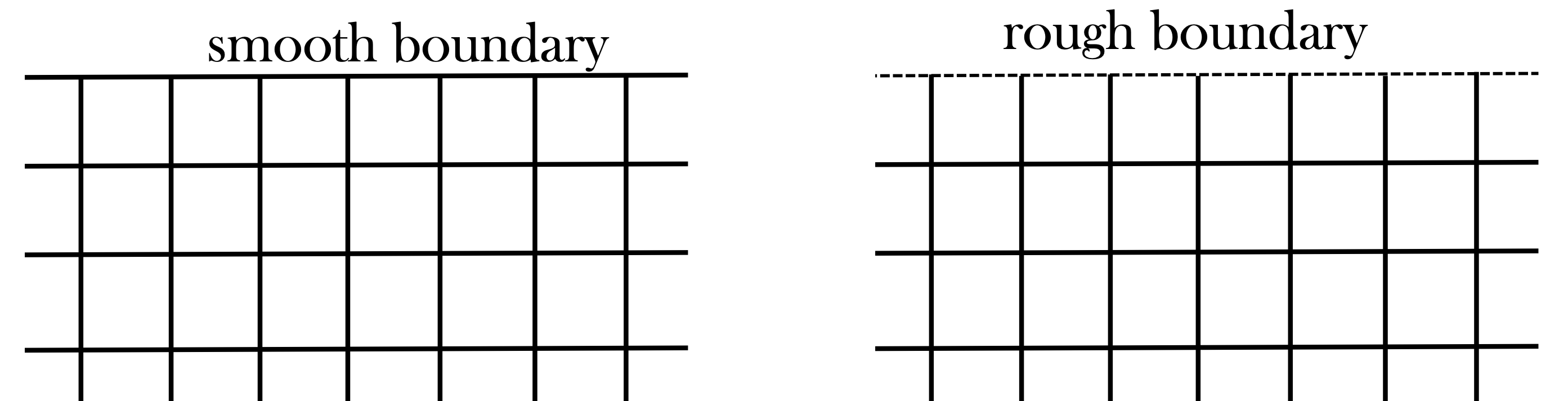}
	\caption{The lattice realizations of smooth and rough boundaries.  When a magnetic charge moves into the smooth boundary, it condenses into the boundary vacuum; when an electric charge moves into the boundary, it condenses into the boundary vacuum. \label{fig:latticebd}}
\end{figure}

For the rough boundary characterized by $K=\IZ_1$, the corresponding Lagrangian algebra is $A_e=\one \oplus e^1\oplus \cdots \oplus e^{p-1}$. When the electric $e^g$ particles move from the bulk to the boundary, they condense into the vacuum of the boundary phase.
The boundary excitations are given by  $\ER=\Mod_{A_e} (\EA(p)) \simeq \EQ(\EA(p),A_e)$.
From string-net perspective, the rough boundary is characterized by a $\Rep(\IZ_p)$-module category $\EM=\Hilb$,  and the boundary excitations are given by $\mathsf{Vect}_{\mathbb{Z}_p} \simeq \mathsf{Fun}_{\mathsf{Rep}(\Zp)} (\mathsf{Hilb},\mathsf{Hilb})$ (the category of all $\Rep(\IZ_p)$-module functors between $\mathsf{Hilb}$ to itself). To summarize, 
\begin{equation}
	\ER=\LMod_{A_e} (\EA(p)) \simeq \EQ(\EA(p),A_e)\simeq \mathsf{Vect}_{\mathbb{Z}_p} \simeq \mathsf{Fun}_{\mathsf{Rep}(\IZ_p)} (\mathsf{Hilb},\mathsf{Hilb} ).
\end{equation}
There are $p$ simple objects in $\ER$, which we denote as $I,M_1,\cdots,M_{p-1}$. Their fusion satisfies $M_a\otimes M_b=M_{a+b}$.

Similarly, for the smooth boundary $K=\IZ_p$, the corresponding Lagrangian algebra is $A_m=\one \oplus m^1\oplus \cdots \oplus m^{p-1}$.
When the magnetic $m^g$ particles move from the bulk to the boundary, they condense into the vacuum of the boundary phase.
The boundary excitations are given by  $\ES=\LMod_{A_m} (\EA(p)) \simeq \EQ(\EA(p),A_m)$.
From string-net perspective, the rough boundary is characterized by a $\Rep(\IZ_p)$-module $\Rep(\IZ_p)$, and
the corresponding topological excitations are given by $\mathsf{Rep}(\mathbb{Z}_p) \simeq \mathsf{Fun}_{\mathsf{Rep}(\IZ_p)} (\mathsf{Rep}(\mathbb{Z}_p),\mathsf{Rep}(\mathbb{Z}_p))$.
To summarize, 
\begin{equation}
	\ES=\LMod_{A_m} (\EA(p)) \simeq \EQ(\EA(p),A_m)\simeq \mathsf{Rep}(\mathbb{Z}_p) \simeq \mathsf{Fun}_{\mathsf{Rep}(\IZ_p)} (\mathsf{Rep}(\mathbb{Z}_p),\mathsf{Rep}(\mathbb{Z}_p)).
\end{equation}
There are $p$ simple objects in $\ES$, which we denote as $I,E_1,\cdots,E_{p-1}$. Their fusion satisfies $E_a\otimes E_b =E_{a+b}$.

Now consider the EM duality symmetry enriched phase $\EA(p)^{\times}_{\IZ_2}$. Under the EM duality symmetry,
\begin{equation}
	\underline{\rho}_1(A_e)=A_m, \quad \underline{\rho}_1(A_m)=A_e.
\end{equation}
This implies that there is no $\IZ_2$ equivariant algebra structure on $A_e$ and $A_m$, thus the symmetry is broken during the anyon condensation.
From fusion rules we presented in Eq.~(\ref{eq:fusionZN}), we see that 
\begin{align}
	 \sigma_k\otimes A_e = \oplus_{g\in \IZ_p} \sigma_g, \forall \,k\in \IZ_p,\\
	\sigma_l\otimes A_m =\oplus_{g\in \IZ_p} \sigma_g, \forall \,l\in \IZ_p.
\end{align}
We denote $\sigma= \oplus_{g\in \IZ_p} \sigma_g$.

Let us first consider the rough boundary phase. The bulk-to-boundary condensation gives
\begin{equation}
	\EA(p)^{\times}_{\IZ_2}\to \ER_{\IZ_2}^{\times}=\ER_0\oplus \ER_1=\{ M_g|g\in \IZ_p\} \oplus \{\sigma\}
\end{equation}
The fusion rule satisfies
\begin{equation}
	M_a\otimes M_b=M_{a+b}, \quad M_a \otimes \sigma =\sigma\otimes M_{a}=\sigma,\quad \sigma\otimes \sigma=\sigma
\end{equation}
for $a,b\in \IZ_p$.
This means that the rough boundary phase is a Tambara-Yamagami category.
For the smooth boundary, the analysis is completely the same and we have
\begin{equation}
	\EA(p)^{\times}_{\IZ_2}\to \ES_{\IZ_2}^{\times}=\ES_0\oplus \ES_1=\{ E_g|g\in \IZ_p\} \oplus \{\sigma\}.
\end{equation}
The fusion rule satisfies
\begin{equation}
	E_a\otimes E_b=E_{a+b}, \quad E_a \otimes \sigma =\sigma\otimes E_{a}=\sigma, \quad \sigma\otimes \sigma=\sigma
\end{equation}
for $a,b\in \IZ_N$. Thus the smooth boundary phase is also a Tambara-Yamagami category.

The bulk phase can be recovered from the boundary phase by taking the relative center. For completeness, let us briefly recall the construction here.
Notice that for a $G$-graded fusion category $\EB_G^{\times}=\oplus_{g\in G} \EB_g$, each $\EB_g$ is a $\EB_{0}$-bimodule category. The relative center of $\EB_g$ with respect to anyon sector $\EB_0$ is defined as follows:

\begin{definition}
The  objects of the relative center $\mathcal{Z}_{\EB_0}(\EB_g)$ are pairs $(M,\gamma)$, where $M\in \EB_g$ and $\gamma=\{\gamma_x:x\otimes M\to M\otimes x\}_{x\in\EB_0}$ is a natural family of isomorphisms such that the following diagram commutes,
	\begin{equation}
	\begin{split}
		\xymatrix{
			x\otimes (M\otimes y) \ar[r]^{\alpha^{-1}_{x,M,y}}& (x\otimes M)\otimes y \ar[r]^{\gamma_{x}} &( M\otimes x)\otimes y \ar[d]^{\alpha_{M,x,y}} \\
			x\otimes (y\otimes M) \ar[r]_{\alpha^{-1}_{x,y,M}}   \ar[u]^{\gamma_{y}}        &(x\otimes y)\otimes M \ar[r]_{\gamma_{x\otimes y}}  &  M\otimes (x\otimes y) }
	\end{split}
\end{equation}
The relative center of $\EB_G^{\times}$ satisfies $\mathcal{Z}_{\EB_0} (\EB_G^{\times}) =\oplus_{g\in G}  \mathcal{Z}_{\EB_0}(\EB_g) $.
\end{definition}

The bulk SET phase $\EA(p)^{\times}_{\IZ_2}$ is in fact the relative centers of $\ER_{\IZ_2}^{\times}$ and $\ES_{\IZ_2}^{\times}$, and 
the rough boundary SET phase and smooth boundary SET phase are Morita equivalent to each other.
To be more precise, we have
\begin{equation}
	\begin{split}
		\xymatrix{
			\EA(p)^{\times}_{\IZ_2} \ar[r]^{\bullet\otimes A_e} \ar[d]_{\bullet\otimes A_m}  & \ER_{\IZ_2}^{\times}  \ar[d]^{\mathcal{Z}_{\ER_0}} \\
		\ES_{\IZ_2}^{\times} \ar@{<->}[ru]_{\text{Morita}}    \ar[r]_{\mathcal{Z}_{\ES_0}}     & \EA(p)^{\times}_{\IZ_2}
	}
	\end{split}
\end{equation}
This commutative diagram summarizes the boundary-bulk duality for the SET phase $\EA(p)^{\times}_{\IZ_2}$.

\section{Lattice realization of EM duality for Abelian phase with dislocation}
\label{sec:latticeReal}

Now let us consider the quantum double lattice realization of these EM duality symmetry defects for the Abelian topological phases.
The EM duality symmetry defect can be realized by dislocation. An explicit example of the construction for the toric code model is given by Kitaev and Kong in \cite{Kitaev2012a}.
Here we will give a generalization to the cyclic Abelian group and discuss the corresponding ribbon operator and topological excitations. The generalization to arbitrary finite Abelian groups is straightforward.

Let us now consider how to construct the dislocation model of $\mathbb{Z}_N$ quantum double phase and how to use the ribbon operator to realize the EM duality symmetry.
Recall that the grading group now is $\mathbb{Z}_2$, with $1\in \mathbb{Z}_2$ corresponding to the EM dual symmetry.

\vspace{1em}
\emph{Lattice model of dislocation.} ---
The process works as follows: for a given $\IZ_N$ lattice gauge theory $H_0$, a pair of defects carrying symmetry fluxes $g$ and $g^{-1}$ can be created and localized at two distant faces by modifying the original Hamiltonian with dislocation.
We can introduce a dislocation line into the original lattice such that the symmetry defects are located at two endpoints of the dislocation line. 
The face operators and vertex operators along the dislocation line need to be modified and two new face operators  $Q^g$ and $Q^{g^{-1}}$ at the endpoints of the dislocation line will be introduced, see Fig.~\ref{fig:lattice} for an illustration.

\begin{figure}[t]
	\centering
	\includegraphics[width=8cm]{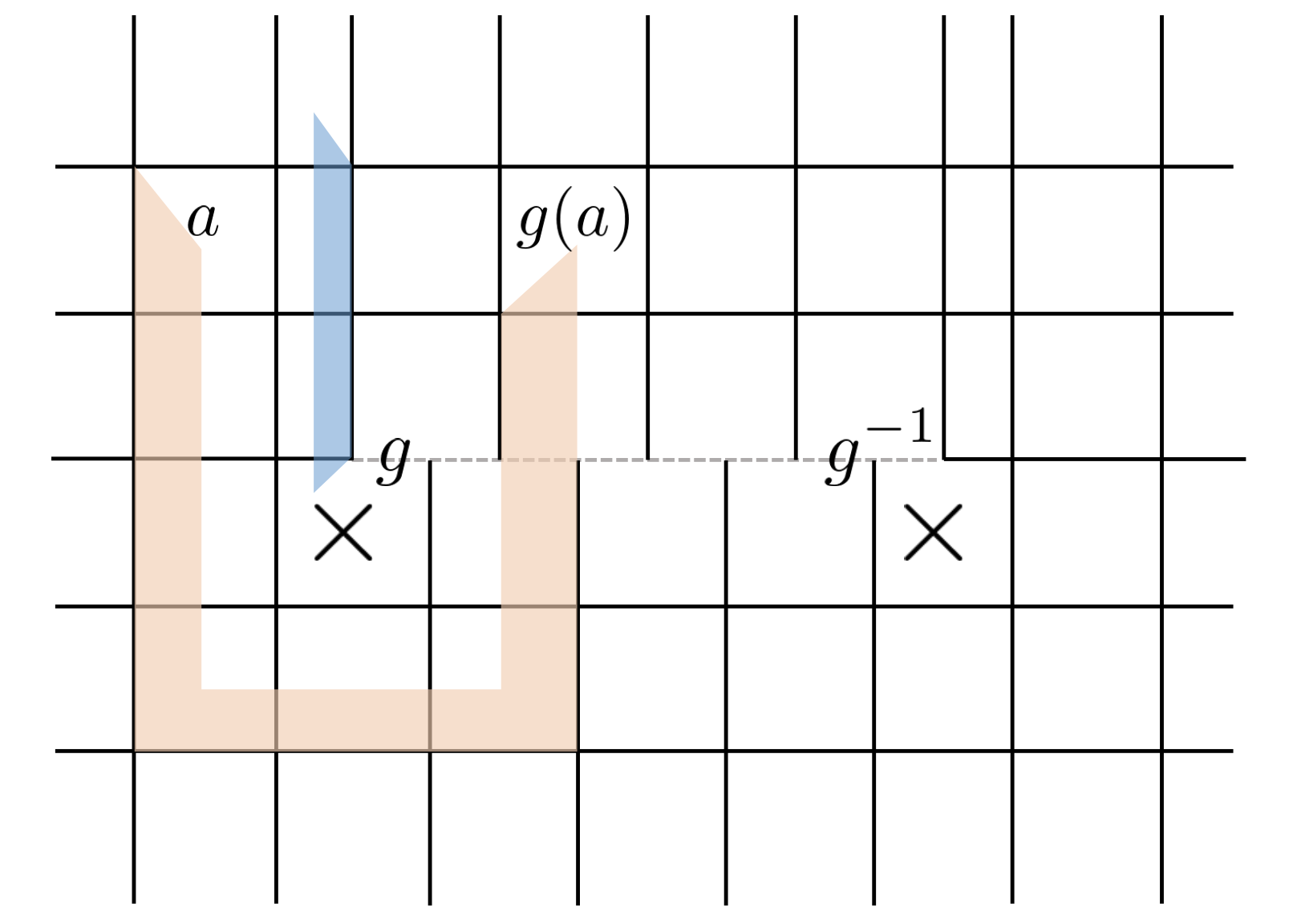}
	\caption{The lattice dislocation realization of SET phase. The face operators along the dislocation line are modified and at two endpoints of the dislocation line, two new face operators are introduced. When a particle $a$ moves around the twist defect carrying the $g$-flux, it becomes $g(a)$. When the anyon $\varepsilon^{g,g}$ moves into defect face, it will be absorbed. \label{fig:lattice}}
\end{figure}

Along the dislocation line, the face operators become
\begin{equation}
\begin{aligned}
	\begin{tikzpicture}
		\draw[-latex,black] (0,0.5) -- (0,1.3);
		\draw[dotted,black] (-0.5,0.5) -- (0.5,0.5);
		\draw[-latex,black] (-0.5,-0.5) -- (-0.5,0.5); 
		\draw[-latex,black] (0.5,-0.5) -- (0.5,0.5); 
		\draw[-latex,black] (-0.5,-0.5) -- (0.5,-0.5); 
		\draw [fill = black] (0,0) circle (1.2pt);
		\node[ line width=0.2pt, dashed, draw opacity=0.5] (a) at (0.75,0){$x_1$};
		\node[ line width=0.2pt, dashed, draw opacity=0.5] (a) at (-0.75,0){$x_3$};
		\node[ line width=0.2pt, dashed, draw opacity=0.5] (a) at (0,-0.7){$x_4$};
		\node[ line width=0.2pt, dashed, draw opacity=0.5] (a) at (0.3,0.9){$x_2$};
	\end{tikzpicture}
\end{aligned}  \quad
\begin{aligned}
	\tilde{B}_f=\frac{1}{N}\sum_{h\in \IZ_N}  (Z_1^{\dagger} X_2^{\dagger} Z_3 Z_4^{\dagger})^h.
\end{aligned} 
\end{equation}
For two endpoints of the dislocation line, we have (notice $g=g^{-1}=1$):
\begin{equation}
	\begin{aligned}
		\begin{tikzpicture}
			\draw[-latex,black] (0,0.5) -- (0,1.3);
			\draw[dotted,black] (0,0.5) -- (0.5,0.5);
			\draw[-latex,black] (-0.5,0.5) -- (0,0.5);
			\draw[-latex,black] (-0.5,-0.5) -- (-0.5,0.5); 
			\draw[-latex,black] (0.5,-0.5) -- (0.5,0.5); 
			\draw[-latex,black] (-0.5,-0.5) -- (0.5,-0.5); 
			\draw [fill = black] (0,0) circle (1.2pt);
			\node[ line width=0.2pt, dashed, draw opacity=0.5] (a) at (0.75,0){$x_1$};
			\node[ line width=0.2pt, dashed, draw opacity=0.5] (a) at (-0.75,0){$x_4$};
			\node[ line width=0.2pt, dashed, draw opacity=0.5] (a) at (0,-0.7){$x_5$};
			\node[ line width=0.2pt, dashed, draw opacity=0.5] (a) at (0.3,0.9){$x_2$};
			\node[ line width=0.2pt, dashed, draw opacity=0.5] (a) at (-0.3,0.7){$x_3$};
		\end{tikzpicture}
	\end{aligned} \quad
Q^1_{\partial_i l}= \frac{1}{N}\sum_{h\in\mathbb{Z}_N} (Z_1^{\dagger}X_2^{\dagger}Y_3Z_4Z_5^{\dagger})^h.
\end{equation}
Here $Y=XZ$ is the Weyl operator (the choice of operator here depends on the direction of the edge $x_3$). 
Similar to the vertex and face operators, we denote $\tilde{B}(f)=Z_1^{\dagger} X_2^{\dagger} Z_3 Z_4^{\dagger}$ and $Q^1(\partial_i l) =Z_1^{\dagger}X_2^{\dagger}Y_3Z_4Z_5^{\dagger}$, then $\tilde{B}_f=(1/N)\sum_{h\in\mathbb{Z}_N}(\tilde{B}(f))^h$ and $Q^1_{\partial_i l} =(1/N)\sum_{h\in \IZ_N} (Q^1(\partial_i l) )^h$.
We see that all of the above local operators commute with each other and commute with all vertex and face operators.
Thus, with a dislocation line $l$, the Hamiltonian becomes
\begin{equation}\label{eq:dislocationH}
	H=H_0-\sum_{f\in l}  \tilde{B}_f -Q^{1}_{\partial_0 l} -Q^{1}_{\partial_1l},
\end{equation}
where $\partial_i l$ represents the faces at the endpoints of the dislocation line $l$, and $H_0$ is the bulk Hamiltonian that does not contain the dislocation line. 

\begin{proposition}
    The ground state space of the model is given by
    \begin{equation}
    \begin{aligned}
            \mathcal{V}_{\rm GS}=&\{|\psi\rangle\,|\,A(v)|\psi\rangle=|\psi\rangle=B(f)|\psi\rangle, \forall\, v,f;\\
            &Q^1(\partial_0 l)|\psi\rangle=|\psi\rangle=Q^1(\partial_1 l)|\psi\rangle; \tilde{B}(f)|\psi\rangle = |\psi\rangle, \forall\, f\in l\}. 
    \end{aligned}
    \end{equation}
\end{proposition}

This can be proved in a similar way as that of Proposition~\ref{prop:GSspace}.
Notice that all local operators, $A(v)$, $B(f)$, $\tilde{B}(f)$, $Q^1(\partial_0 l)$ and $Q^1(\partial_1 l)$ are unitary operators whose eigenvalues are $\omega_N^g$, $g=0,\cdots,N-1$.

\vspace{1em}
\emph{EM-exchange ribbon operator.} ---
Let us now consider how to construct the ribbon operator across the defect line
that realizes the EM duality symmetry.
Under the EM duality, the electric and magnetic charges will be exchanged:
\begin{equation}
    e\leftrightarrow m, \quad \varepsilon^{a,b}\leftrightarrow \varepsilon^{b,a}.
\end{equation}
We will see that when the bulk topological excitations pass through the dislocation line in either way, their electric charge and magnetic charge will be exchanged.

Let us first recall that a ribbon $\rho$ consists of direct and dual triangles. The EM-exchange ribbon is the one crossing the dislocation line, see Fig.~\ref{fig:EM_ribbon}. In this case, the EM-ribbon operator is defined by 
\begin{equation}
    F_\rho^{g,h} = (Z_1^\dagger)^g (X_2)^h(Z_3^\dagger)^g (Z_4^\dagger)^h (X_5)^g. 
\end{equation}
At one end of the ribbon, an anyon $\varepsilon^{g,h}=e^g\otimes m^h$ is created, when this anyon is dragged across the dislocation line, an EM duality operation will act on this anyon and it becomes $\varepsilon^{h,g}=e^h\otimes m^g$.

\begin{figure}[t]
	\centering
	\includegraphics[width=6cm]{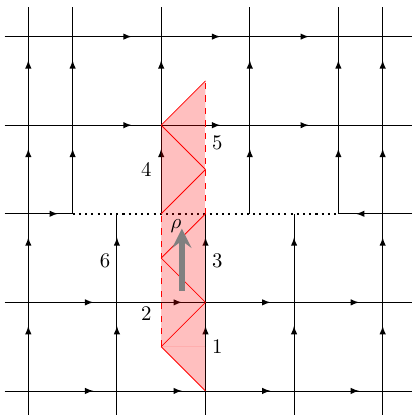}
	\caption{An EM-ribbon $\rho$ crossing the dislocation line. When passing the dislocation line, an $e$-particle is changed to an $m$-particle, and an $m$-particle is changed to an $e$-particle. \label{fig:EM_ribbon}}
\end{figure}

\begin{proposition}
   The EM-exchange ribbon operator satisfies $[F_{\rho},A(v)]=0$, $[F_{\rho},B(f)]=0$, and $[F_{\rho},\tilde{B}(f)]=0$ for all vertices, faces that are not the ends of the ribbon.
\end{proposition}

\begin{proof}
  One only needs to check that $[F_\rho,\tilde{B}(f)]=0$ for $f\in l$. Consider the configuration in Fig.~\ref{fig:EM_ribbon}, where $\tilde{B}(f) = Z_3^\dagger X_4^\dagger Z_6Z_2^\dagger$. As $(X_2)^hZ_2^\dagger = \omega_N^{-h}Z_2^\dagger(X_2)^h$ and $(Z_4^\dagger)^hX_4^\dagger = \omega_N^hX_4^\dagger (Z_4^\dagger)^h$, one has $F_\rho^{g,h}\tilde{B}(f)  = \omega_N^{-h}\omega_N^h\tilde{B}(f)F^{g,h}_\rho=\tilde{B}(f)F_\rho^{g,h}$. 
\end{proof}

The above proposition guarantees that no other topological excitations are created during the EM exchange process.
If we regard the end of the dislocation line as a symmetry defect, then the EM exchange ribbon realizes the $\mathbb{Z}_2$-crossed braiding given in Eq.~\eqref{eq:Gbraiding}.
The EM exchange ribbon can drag an anyon $a=\varepsilon^{g,h}$ around the symmetry defect. When crossing the dislocation line, the $Z^g$-string becomes $X^g$ string, thus the $e^g$ particle becomes $m^g$ particle. Similarly, the $X^h$-string becomes $Z^h$ string, thus the $m^h$ particle becomes $e^h$ particle. This is exactly the crossed braiding $ \sigma_k\otimes a= \underline{\rho}_1(a) \otimes \sigma_k$ given by the EM duality symmetry.

Another result derived from the EM exchange ribbon construction is that we can only construct a double-braiding closed ribbon across the dislocation line. See Fig.~\ref{fig:closed_ribbon}. If a ribbon crosses the dislocation line by an odd number of times, the anyon's electric and magnetic charges are exchanged, and we can never construct such a closed ribbon. On the other hand, if a ribbon crosses the dislocation line by an even number of times, the anyon remains unchanged, and thus such a closed ribbon can be constructed.

\begin{figure}[t]
    \centering
    \subfloat[Incompatibility at the initial/terminal site]{
        \includegraphics[width=6.5cm]{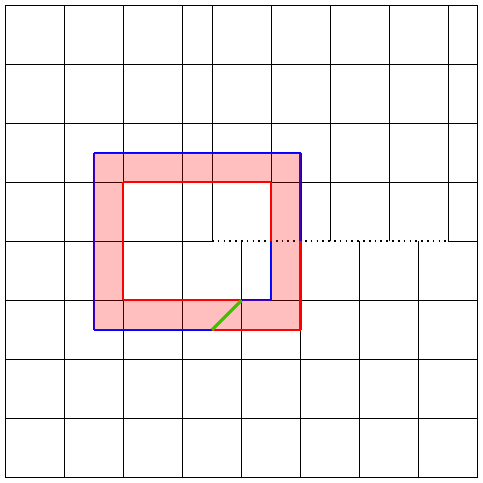}
    }%
    \qquad
    \subfloat[Compatible EM-exchange closed ribbon]{
        \includegraphics[width=6.5cm]{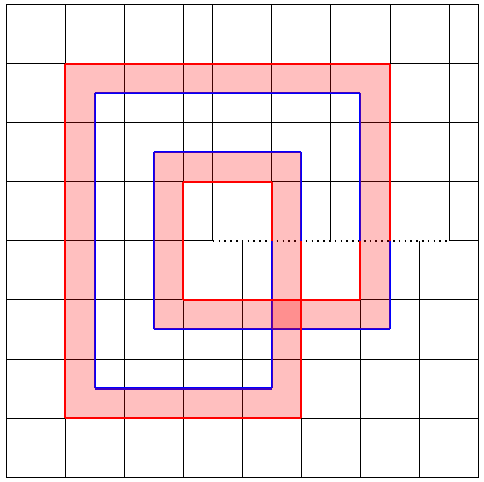}
    }%
    \caption{(a) This type of EM-exchange closed ribbon is not possible because it is incompatible at the green site (the initial/terminal site). (b) This is the only possible type of EM-exchange closed ribbon; it must cross the dislocation line twice (or an even number of times).}
    \label{fig:closed_ribbon}
\end{figure}

\vspace{1em}
\emph{Absorbing anyons via dislocation.} ---
A crucial property of the two ends of the dislocation line is that they behave like a non-Abelian anyon.
Each end can be regarded as a zero-dimensional defect that connects a smooth domain wall and a rough domain wall.
We can thus consider the fusion and braiding of bulk anyons with these defects and also the fusion and braiding between these defects.
An interesting phenomenon is that bulk anyons $\varepsilon^{g,g}$ will be absorbed by the defect. This has been proven in Eq.~\eqref{eq:fusionZN} via the algebraic analysis; here we show that in the lattice realization, this can be realized via the ribbon operators that terminate at the defect site.

From the fusion rule in Eq.~(\ref{eq:fusionZN}), we see that anyons that are invariant under EM duality symmetry will be absorbed when fuse with symmetry defects:
\begin{equation}
\varepsilon^{g,g} \otimes \sigma_k=\sigma_k, \quad \forall g\in \IZ_N.
\end{equation}
Thus the symmetry defects can be regarded as sinks of these anyons.
These fusion processes can be realized by the blue ribbon operator $F^{g,g}_\xi$ drawn in Fig.~\ref{fig:lattice}. Moving the anyon $\varepsilon^{g,g} $ is realized by a $Z^g$-string and an $X^g$ string along the direct and dual boundaries of the ribbon $\xi$. Since $F^{g,g}_\xi$ commutes with all vertex and face operators near the defects, the anyon $\varepsilon^{g,g} $ is absorbed by the defect.

\vspace{1em}
\emph{Boundary lattice model.} ---
The lattice realization of the rough and smooth boundaries can be obtained by modifying the boundary local operators. As shown in Fig.~\ref{fig:latticebd}, the face operators near the smooth boundary remain unchanged, but the vertex operators become
\begin{equation}
		\begin{aligned}
		\begin{tikzpicture}
			\draw[-latex,black,line width=0.2pt] (1,0.5) -- (0,0.5);
			\draw[-latex,black,line width=0.2pt] (0,0.5) -- (-1,0.5);
			\draw[-latex,black] (0,-0.5) -- (0,0.5); 
			\draw [fill = black] (0,0.5) circle (1.2pt);
			\node[ line width=0.2pt, dashed, draw opacity=0.5] (a) at (-0.3,-0.1){$x_3$};
			\node[ line width=0.2pt, dashed, draw opacity=0.5] (a) at (0.5,0.7){$x_1$};
			\node[ line width=0.2pt, dashed, draw opacity=0.5] (a) at (-0.5,0.7){$x_2$};
		\end{tikzpicture}
	\end{aligned}\quad
\bar{A}_v=\frac{1}{N}\sum_{h\in \IZ_N} (X_1X_2^{\dagger} X_3)^h.
\end{equation}
Similarly, for the rough boundary, the vertex operators remain unchanged, but the face operators become
\begin{equation}
	\begin{aligned}
		\begin{tikzpicture}
			\draw[dotted,black,line width=0.2pt] (0.5,0.5) -- (-0.5,0.5);
			\draw[-latex,black] (0.5,-0.5) -- (0.5,0.5); 
			\draw[-latex,black] (-0.5,-0.5) -- (-0.5,0.5); 
			\draw[-latex,black] (0.5,-0.5) -- (-0.5,-0.5); 
			\draw [fill = black] (0,0) circle (1.2pt);
			\node[ line width=0.2pt, dashed, draw opacity=0.5] (a) at (0,-0.3){$x_3$};
			\node[ line width=0.2pt, dashed, draw opacity=0.5] (a) at (0.9,0){$x_1$};
			\node[ line width=0.2pt, dashed, draw opacity=0.5] (a) at (-0.9,0){$x_2$};
		\end{tikzpicture}
	\end{aligned}\quad
	\bar{B}_v=\frac{1}{N}\sum_{h\in \IZ_N} (Z_1^{\dagger}Z_2 Z_3)^h.
\end{equation}
The condensed symmetry defect $\sigma$ can be modeled by a dislocation line which connects boundary $\sigma$ with a bulk symmetry defect $\sigma_k$.
The bulk to boundary condensation is described by the bulk-to-boundary ribbon operators.

\vspace{1em}
\emph{EM duality.} ---
We have presented a detailed discussion of the EM duality symmetry. In this part, let us construct a lattice realization of the EM duality for the SET phase.

\begin{figure}[b]
	\centering
	\includegraphics[width=7cm]{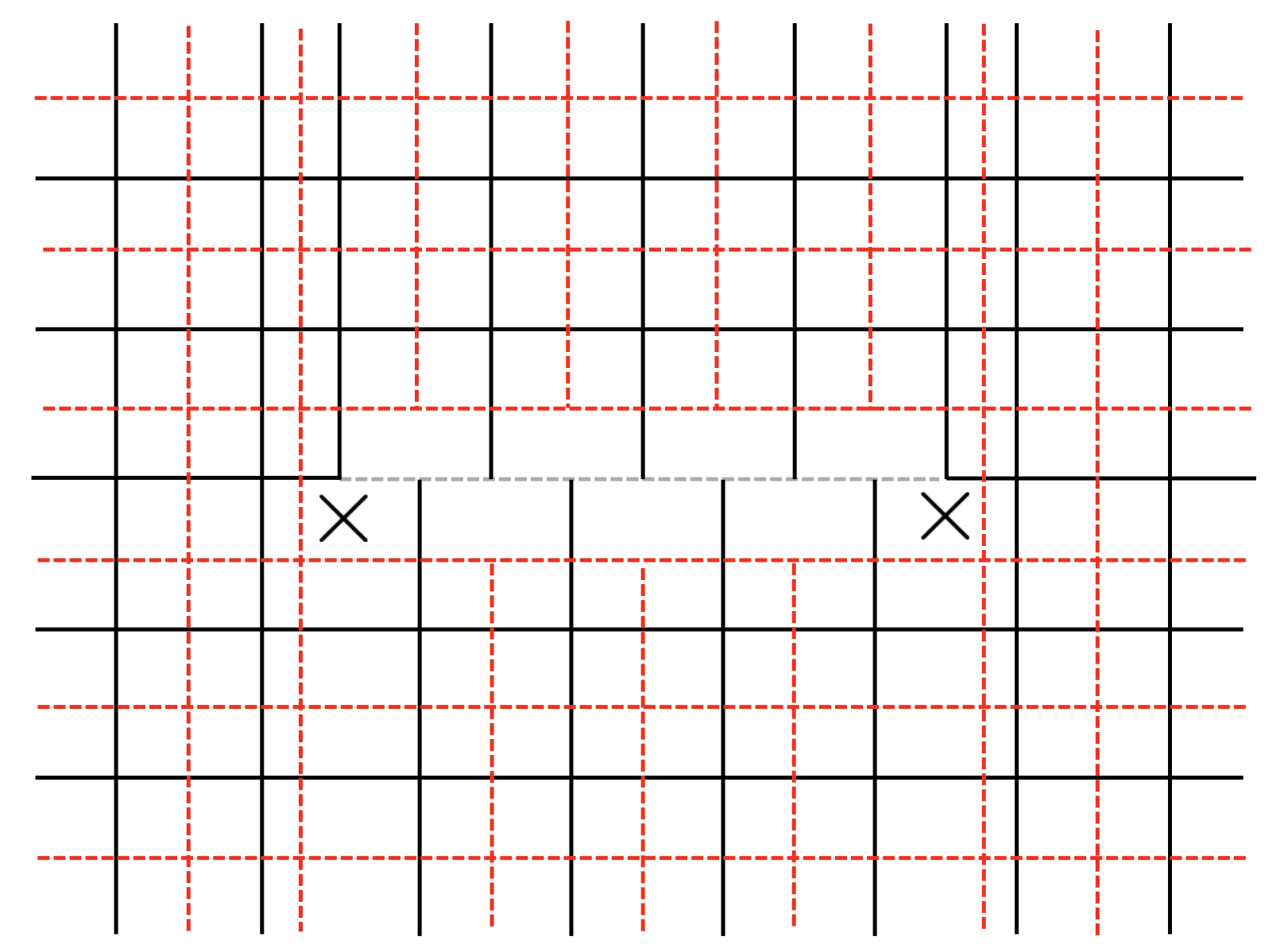}
	\caption{The direct and dual lattice of the symmetry enriched quantum double model. \label{fig:dual}}
\end{figure}
Mathematically speaking, the EM duality of the quantum double model stems from the Poincar\'{e} duality of the underlying space manifold.
For a given lattice $\Sigma=V\cup E\cup F$, we introduce its dual lattice $\Sigma^*$ for which the vertices are the faces of $\Sigma$ and the faces are the vertices of $\Sigma$. Since each edge of $\Sigma$ has a corresponding direction, the direction of the corresponding dual edge $e^*$ is obtained by rotating by $\pi/2$ of the direct edge counterclockwise.

We can introduce a local unitary map $U_e$ between $\mathcal{H}_e=\mathbb{C}[\IZ_N]$ and $\mathcal{H}_{e^*}=F(\IZ_N)=\mathbb{C}[\hat{\IZ}_N]$ with $\hat{\IZ}_N\simeq \IZ_N$ the dual group of $\IZ_N$. 
The local unitary operator is the generalized Hadamard operator 
\begin{equation}
U_e=\sum_{g\in \IZ_N} |\psi_g\rangle \langle g|
\end{equation}
where $|\psi_g\rangle$ is the eigenstate of $X_N$ as in Eq.~(\ref{eq:Xeigen}). Notice that 
\begin{equation}
	\begin{aligned}
	U_eZU_e^{\dagger}=Z^{*},\quad 
	U_eXU_e^{\dagger}=X^{*},
	\end{aligned}
\end{equation}
where $Z^*=\sum_{h\in \IZ_N} w^h |\psi_h\rangle \langle \psi_h| $ and $X^*=\sum_{h\in\IZ_N} | \psi_{h+1}\rangle \langle \psi_h| $ are the Weyl operators in the dual space basis.
Then the global unitary operator $U_{\Sigma}=\otimes_{e\in E} U_e$ gives a duality between quantum double models $D(\IZ_N)_{\Sigma}$ on $\Sigma$ and $D(\hat{\IZ}_N)_{\Sigma^*}$ on $\Sigma^*$:
\begin{equation}
	U_{\Sigma} A^{\Sigma}_v U_{\Sigma}^{\dagger} = B_v^{\Sigma^*},\quad U_{\Sigma} B^{\Sigma}_f U_{\Sigma}^{\dagger} = A_f^{\Sigma^*}.
\end{equation}
Now for the SET phase $\AN$ with boundaries, all of the bulk symmetry defect operators, face and vertex operators along the dislocation line, and the boundary operators are transformed in a similar way.
See Fig. \ref{fig:dual} for an illustration of the direct and dual lattices of SET.
From this EM duality, the electric charge on the vertex of $\Sigma$ becomes a magnetic charge on the face. When the model has a gapped smooth boundary, the EM duality maps it into a rough boundary, and similarly, a rough boundary is mapped to a smooth boundary.
Thus EM duality allows us to understand the smooth boundary and rough boundary from each other's perspective.

For non-Abelian group $G$, there is no EM duality, but the partial EM duality for some Abelian subgroups can be constructed \cite{hu2020electric}.
The more general setting is for the Hopf algebraic model, in which the EM duality can be established using the Tannaka-Krein duality \cite{Buerschaper2013a}.
Our construction here provides a simple example for further investigation of the extended Tannaka-Krein duality, which is applied for the $G$-crossed UBFC together with a corresponding Lagrangian algebra $A$. This general case will be discussed in our other work.

The EM duality provides a way to understand the direct lattice model $D(\IZ_N)_{\Sigma}$ and dual lattice model $D(\IZ_N)_{\Sigma^*}$ from each other's perspective.
For example, if we obtain the ground state $|\Omega\rangle_{\Sigma}$ of the direct lattice, our explicit construction of EM duality allows us to obtain the ground state of the dual lattice directly from the unitary transformation $|\Omega\rangle_{\Sigma^*}=U_{\Sigma} |\Omega\rangle_{\Sigma}$. For excited states, using the duality for ribbon operators, we also have this correspondence.









\section{Conclusion}

In this work, we systematically investigated the simplest class of SET phases, EM duality symmetry enriched cyclic Abelian topological phase $\AN$.
The phase is studied via both algebraic analysis and lattice Hamiltonian realizations.
The $\IZ_N$ quantum double phase has EM duality symmetry, a $\IZ_2$ symmetry, which is a special case of the more general categorical symmetry. Using this symmetry, the $\EA(N)$ phase can be enriched into a SET phase.
By introducing the anyon condensation process for the SET phase, the gapped boundaries are discussed in detail.
From the boundary-bulk duality discussion, we show that the SET phase  $\AN$ is characterized by the relative center of Morita equivalent Tambara-Yamagami categories, which correspond to different gapped boundaries.
We also present an explicit lattice relation of this SET phase and its gapped boundaries and the explicit lattice realization of the EM duality is constructed.

The general lattice Hamiltonian constructions of SET phase are still largely unexploited. There are some constructions for special examples of SET phases \cite{barkeshli2020reflection,barkeshli2020relative,wang2021exactly}.
The symmetry-enriched string-net model is discussed in \cite{Heinrich2016symmetry}, which is conjectured to realize a large class of SET phases. There are also some trials for introducing new tools for studying the SET phase, including tensor network representation and neural network representation of topological quantum states of the SET phases \cite{williamson2017symmetryenriched,Bridgeman2017tensor,jia2019quantum,jia2018efficient,zhang2018efficient}.
However, a systematic construction based on Kitaev quantum double model for arbitrary group $G$ (or more generally, Hopf algebra $H$) is still an open problem; our result is an example in this direction. We conjecture that all Kitaev models based on semisimple Hopf algebras can be enriched by EM duality symmetry, and we left the explicit construction of these models for future study.

Besides the fundamental importance, investigating the SET phase is also crucial for quantum information processing, especially for topological quantum computation and topological quantum memory.
We know that the real samples of quantum materials always have boundaries and defects, hence a better understanding of these defects can help us to design more robust topological quantum memory materials.
Another crucial point is that by introducing the defects into the Abelian phase, the resulting SET phase can be used to do the universal topological quantum computation, which cannot be done by the original Abelian phase \cite{Cong2017universal,Fowler2012,Brown2017poking}.
We leave the study of these possible applications for our future work.

\subsection*{Acknowledgments}
Z.J. would like to acknowledge Liang Kong and Zhenghan Wang for their helpful discussions. Special thanks to the math department of UCSB, where part of this work was written.
Z.J. and D.K. are supported by National Research Foundation in Singapore and A*STAR under its CQT Bridging Grant.
S.T. was partially supported by NSF grant DMS-2100288 and by Simons Foundation Collaboration Grant for Mathematicians \#580839 during his PhD study at Purdue University, and now is supported by postdoc fund from BIMSA. He would like to thank his advisor Uli Walther for the generosity and patience.

\appendix

\section{Tambara-Yamagami category}
\label{sec:TYcat}
In this section, let us review the definition of the Tambara-Yamagami category $\mathsf{TY}(G,\chi,\nu)$ associated with an Abelian group $G$, a symmetric non-degenerate bicharacter $\chi:G\times G\to \mathbb{C}^{\times}$ \footnote{Here $\mathbb{C}^{\times}$ is the set of all invertible elements in $\mathbb{C}$, i.e., $\mathbb{C}^{\times}=\mathbb{C}\setminus\{0\}$.} and a Frobenius-Schur indicator $\nu=\pm 1$ \cite{tambara1998tensor}.
By a bicharacter we mean a bilinear function $\chi:G\times H\to \mathbb{C}^{\times}$, \emph{viz.},
\begin{equation}
    \begin{aligned}
\chi(g\cdot h,k)=\chi(g,k)\chi(h,k),\\
  \chi(g,k\cdot l)=\chi(g,k)\chi(g,l),
    \end{aligned}
\end{equation}
for all $g,h\in G$ and $k,l\in H$.
A bicharacter $\chi:G\times G \to \mathbb{C}^{\times}$ is called symmetric if $\chi(g,h)=\chi(h,g)$ for all $g,h\in G$.
The non-degeneracy of $\chi$ means that $\chi$ is non-degenerate as a bilinear form.

The Tambara-Yamagami category $\mathsf{TY}(G,\chi,\nu)$ is $\mathbb{Z}_2$-graded, meaning that there are two sectors of objects
\begin{equation}
    \operatorname{Obj} \mathsf{TY}(G,\chi,\nu)= \{g|g\in G\}_0\oplus \{\sigma\}_1.
\end{equation}
The fusion rule satisfies the $\mathbb{Z}_2$ grading relation:
\begin{equation}
g\otimes h=g\cdot h,\quad
g\otimes \sigma=\sigma\otimes g =\sigma,\quad
\sigma\otimes \sigma =\oplus_{g\in G} g.
\end{equation}
The only nontrivial $F$-symbols are the following
\begin{align}
   [ F_{\sigma}^{g\sigma h}]_{\sigma}^{\sigma}= [ F_{\sigma}^{\sigma g \sigma}]_{h}^{\sigma}=\chi(g,h),\\
   [F_{\sigma}^{\sigma\sigma\sigma}]_{g}^h=\frac{\nu}{\sqrt{|G|}} \chi(g,h)^{-1}.
\end{align}

\bibliographystyle{apsrev4-1-title}
\bibliography{mybib.bib}

\end{document}